\documentclass{fundam}

\publyear{25}
\papernumber{1}
\volume{194}
\issue{3}
\theDOI{10.46298/fi.10010}

\versionForARXIV

\usepackage{graphicx}
\usepackage{tikz}
\usetikzlibrary{automata, positioning, arrows}

\usepackage{relsize}
\usepackage[utf8]{inputenc}
\usepackage{amsfonts}
\usepackage{amssymb}
\usepackage{amsbsy}
\usepackage{mathrsfs}
\usepackage{enumerate}
\usepackage{amsmath}
\usepackage{url}
\usepackage{color}
\usepackage{graphicx}
\usepackage{cmll}
\usepackage{stmaryrd}
\usepackage{wrapfig}
\usepackage{bm} 
\usepackage{tikz-cd}
\usepackage[all]{xy}
\usepackage[hidelinks]{hyperref}

\usepackage{thmtools,thm-restate}
\usepackage{mathtools}
\usepackage{xspace}
\usepackage{cancel}
\usepackage[normalem]{ulem}

\newcommand{\defeq}{\coloneqq}

\newcommand{\Set}[1]{\{#1\}}

\newcommand{\UnitSub}[1]{\mbox{\it unit}_{#1}}
\newcommand{\Unit}[1]{[#1]}

\newcommand{\Bind}{\star}

\newcommand{\Subst}[2]{[ #1 / #2]}

\newcommand{\ValTerm}{\textit{Val}}
\newcommand{\ComTerm}{\textit{Com}}

\newcommand{\Red}{\red}
\newcommand{\RedStar}{\red^*}

\newcommand{\LeftUnit}{\textit{Left unit}}
\newcommand{\RightUnit}{\textit{Right unit}}
\newcommand{\Assoc}{\textit{Assoc}}

\newcommand{\lc}{\lambda_{\normalfont\scalebox{.4}{\copyright}}}

\newcommand{\II}{\mathbf{I}}

\newcommand{\betac}{\beta_c}

\newcommand{\bigRVal}{(\textit{Val-conv})}
\newcommand{\bigRBind}{(\Bind\textit{-conv})}
\newcommand{\bigRSet}{(\textit{set-conv})}
\newcommand{\bigRGet}{(\textit{get-conv})}

\newcommand{\letsf}{\mathsf{let}}
\newcommand{\insf}{\mathsf{in}}
\newcommand{\Let}[3]{\letsf \, #1 \!\defeq\! #2 \, \insf \, #3}

\newcommand{\Var}{\textit{Var}}

\newcommand{\Sem}[1]{[\hspace{-0.6mm}[ #1 ]\hspace{-0.6mm}]}

\makeatletter
\newcommand{\circlearrow}{}
\DeclareRobustCommand{\circlearrow}{%
	\mathrel{\vphantom{\rightarrow}\mathpalette\circle@arrow\relax}%
}
\newcommand{\circle@arrow}[2]{%
	\m@th
	\ooalign{%
		\hidewidth$#1\circ\mkern1mu$\hidewidth\cr
		$#1\longrightarrow$\cr}%
}
\makeatother

\renewcommand{\int}{\mathsf{i}}

\newcommand{\Inter}{\wedge}

\newcommand{\Th}{\textit{Th}}
\newcommand{\der}{\vdash}

\newcommand{\Env}{\mbox{\it Env}}
\newcommand{\metalambda}{%
	\mathop{%
		\rlap{$\lambda$}%
		\mkern2mu
		\raisebox{.275ex}{$\lambda$}%
	}%
}

\newcommand{\CatC}{{\cal C}}

\newcommand{\CatDom}{\textbf{Dom}}

\usepackage{xcolor}

\newcommand{\RED}[1]{{\color{red}{#1}}}

\newcommand{\violet}[1]{{\color{violet}{#1}}}

\renewcommand{\to}{\xrightarrow{}}

\newcommand{\red}{\rightarrow}

\newcommand{\lam}{\lambda}

\makeatletter
\newcommand{\uset}[3][0ex]{%
	\mathrel{\mathop{#3}\limits_{
			\vbox to#1{\kern-6\ex@
				\hbox{$\scriptstyle#2$}\vss}}}}
\makeatother

\RequirePackage{ifthen}

\newcommand{\version}{0}
\newcommand{\commenti}{0} 
\newcommand{\condinc}[2]{\ifthenelse{\equal{\commenti}{0}}{#1}{**\violet{#2}} }
\newcommand{\SLV}[2]{\ifthenelse{\equal{\version}{0}}{#1}{ \RED{#2}}}
\newcommand{\unit}[1]{[ #1 ]}
\newcommand{\get}[2]{\textit{get}_{#1} (#2)}
\newcommand{\set}[3]{\textit{set}_{#1} (#2, #3)}

\newcommand{\Aget}[1]{\textit{lkp}_{#1}}
\newcommand{\Aset}[1]{\textit{upd}_{#1}}

\newcommand{\StoreSort}{\textit{Store}}
\newcommand{\emp}{\textit{emp}}
\newcommand{\SemD}[1]{\Sem{#1}^D}
\newcommand{\SemSD}[1]{\Sem{#1}^{\GS D}}
\newcommand{\BigStep}[4]{(#1 , #2) \Downarrow (#3 , #4)}

\newcommand{\SmallStep}[4]{(#1 , #2) \Red (#3 , #4)}
\newcommand{\SmallStepStar}[4]{(#1 , #2) \RedStar (#3 , #4)}

\newcommand{\tuple}[1]{\langle #1 \rangle}
\newcommand{\Label}{\textbf{L}}
\newcommand{\dom}[1]{\textit{dom}( #1 )}
\newcommand{\mutelambda}{\lambda \underbar{~} \,}
\newcommand{\storeMap}{\varsigma}

\newcommand{\finMap}{\varsigma}
\newcommand{\omegaD}{\omega_D}

\newcommand{\omegaS}{\omega_{S}}
\newcommand{\omegaC}{\omega_{C}}
\newcommand{\omegaSD}{\omega_{\GS D}}

\newcommand{\leqD}{\leq_D}
\newcommand{\leqS}{\leq_S}
\newcommand{\leqC}{\leq_{C}}
\newcommand{\leqSD}{\leq_{\GS D}}

\newcommand{\Lang}{\mathcal{L}}

\newcommand{\LangD}{\Lang_{D}}
\newcommand{\LangS}{\Lang_{S}}
\newcommand{\LangC}{\Lang_{C}}
\newcommand{\LangSD}{\Lang_{\GS D}}

\newcommand{\SemS}[1]{\Sem{#1}^{\Store}}

\renewcommand{\II}{\mathcal{I}}
\newcommand{\sat}[2]{| #1 |_{#2}}

\newcommand{\seqR}{(\, ;  )}

\renewcommand{\II}{{\cal I}}

\newcommand{\omegaR}{(\omega)}
\newcommand{\interR}{(\Inter)}
\newcommand{\leqR}{(\leq)}
\newcommand{\varR}{(\textit{var})}
\newcommand{\lambdaR}{(\lambda)}
\newcommand{\unitR}{(\textit{unit})}
\newcommand{\bindR}{(\Bind)}
\newcommand{\getR}{(\textit{get})}
\newcommand{\setR}{(\textit{set})}
\newcommand{\lkpR}{(\Aget{})}
\newcommand{\updRa}{(\textit{upd}_1)}
\newcommand{\updRb}{(\textit{upd}_2)}
\newcommand{\confR}{(\textit{conf})}

\newcommand{\betaR}{(\betac)}
\newcommand{\BindR}{(\Bind\textit{-red})}
\newcommand{\getredR}{(\textit{get-red})}
\newcommand{\setredR}{(\textit{set-red})}

\newcommand{\myparagraph}[1]{\medskip \noindent {\bf #1}.}
\newcommand{\lamImp}{\lambda_{imp}}
\newcommand{\Dcat}{{\cal D}}

\newcommand{\GS}{\mathbb{S}}

\newcommand{\GSext}[1]{#1^{\GS}}
\newcommand{\Store}{S}

\newcommand{\SemC}[1]{\Sem{#1}^C}

\renewcommand{\II}{ }
\newcommand{\BotC}{\spadesuit}

\newcommand{\Iff}{\Leftrightarrow}
\newcommand{\Then}{\Rightarrow}

\def \QED {~\hfill\hbox{\rule{5.5pt}{5.5pt}\hspace*{.5\leftmarginii}}}

\usepackage{ stmaryrd }
\newdimen\proofrulebreadth \proofrulebreadth=.05em
\newdimen\proofdotseparation \proofdotseparation=1.25ex
\newdimen\proofrulebaseline \proofrulebaseline=2ex
\newcount\proofdotnumber \proofdotnumber=3
\let\then\relax
\def\hfi{\hskip0pt plus.0001fil}
\mathchardef\squigto="3A3B
\newif\ifinsideprooftree\insideprooftreefalse
\newif\ifonleftofproofrule\onleftofproofrulefalse
\newif\ifproofdots\proofdotsfalse
\newif\ifdoubleproof\doubleprooffalse
\let\wereinproofbit\relax
\newdimen\shortenproofleft
\newdimen\shortenproofright
\newdimen\proofbelowshift
\newbox\proofabove
\newbox\proofbelow
\newbox\proofrulename
\def\shiftproofbelow{\let\next\relax\afterassignment\setshiftproofbelow\dimen0 }
\def\shiftproofbelowneg{\def\next{\multiply\dimen0 by-1 }%
\afterassignment\setshiftproofbelow\dimen0 }
\def\setshiftproofbelow{\next\proofbelowshift=\dimen0 }
\def\setproofrulebreadth{\proofrulebreadth}

\def\prooftree{
\ifnum  \lastpenalty=1
\then   \unpenalty
\else   \onleftofproofrulefalse
\fi
\ifonleftofproofrule
\else   \ifinsideprooftree
        \then   \hskip.5em plus1fil
        \fi
\fi
\bgroup
\setbox\proofbelow=\hbox{}\setbox\proofrulename=\hbox{}%
\let\justifies\proofover\let\leadsto\proofoverdots\let\Justifies\proofoverdbl
\let\using\proofusing\let\[\prooftree
\ifinsideprooftree\let\]\endprooftree\fi
\proofdotsfalse\doubleprooffalse
\let\thickness\setproofrulebreadth
\let\shiftright\shiftproofbelow \let\shift\shiftproofbelow
\let\shiftleft\shiftproofbelowneg
\let\ifwasinsideprooftree\ifinsideprooftree
\insideprooftreetrue
\setbox\proofabove=\hbox\bgroup$\displaystyle 
\let\wereinproofbit\prooftree
\shortenproofleft=0pt \shortenproofright=0pt \proofbelowshift=0pt
\onleftofproofruletrue\penalty1
}

\def\eproofbit{
\ifx    \wereinproofbit\prooftree
\then   \ifcase \lastpenalty
        \then   \shortenproofright=0pt  
        \or     \unpenalty\hfil         
        \or     \unpenalty\unskip       
        \else   \shortenproofright=0pt  
        \fi
\fi
\global\dimen0=\shortenproofleft
\global\dimen1=\shortenproofright
\global\dimen2=\proofrulebreadth
\global\dimen3=\proofbelowshift
\global\dimen4=\proofdotseparation
\global\count255=\proofdotnumber
$\egroup  
\shortenproofleft=\dimen0
\shortenproofright=\dimen1
\proofrulebreadth=\dimen2
\proofbelowshift=\dimen3
\proofdotseparation=\dimen4
\proofdotnumber=\count255
}

\def\proofover{
\eproofbit 
\setbox\proofbelow=\hbox\bgroup 
\let\wereinproofbit\proofover
$\displaystyle
}%
\def\proofoverdbl{
\eproofbit 
\doubleprooftrue
\setbox\proofbelow=\hbox\bgroup 
\let\wereinproofbit\proofoverdbl
$\displaystyle
}%
\def\proofoverdots{
\eproofbit 
\proofdotstrue
\setbox\proofbelow=\hbox\bgroup 
\let\wereinproofbit\proofoverdots
$\displaystyle
}%
\def\proofusing{
\eproofbit 
\setbox\proofrulename=\hbox\bgroup 
\let\wereinproofbit\proofusing
\kern0.3em$
}

\def\endprooftree{
\eproofbit 
  \dimen5 =0pt
\dimen0=\wd\proofabove \advance\dimen0-\shortenproofleft
\advance\dimen0-\shortenproofright
\dimen1=.5\dimen0 \advance\dimen1-.5\wd\proofbelow
\dimen4=\dimen1
\advance\dimen1\proofbelowshift \advance\dimen4-\proofbelowshift
\ifdim  \dimen1<0pt
\then   \advance\shortenproofleft\dimen1
        \advance\dimen0-\dimen1
        \dimen1=0pt
        \ifdim  \shortenproofleft<0pt
        \then   \setbox\proofabove=\hbox{%
                        \kern-\shortenproofleft\unhbox\proofabove}%
                \shortenproofleft=0pt
        \fi
\fi
\ifdim  \dimen4<0pt
\then   \advance\shortenproofright\dimen4
        \advance\dimen0-\dimen4
        \dimen4=0pt
\fi
\ifdim  \shortenproofright<\wd\proofrulename
\then   \shortenproofright=\wd\proofrulename
\fi
\dimen2=\shortenproofleft \advance\dimen2 by\dimen1
\dimen3=\shortenproofright\advance\dimen3 by\dimen4
\ifproofdots
\then
        \dimen6=\shortenproofleft \advance\dimen6 .5\dimen0
        \setbox1=\vbox to\proofdotseparation{\vss\hbox{$\cdot$}\vss}%
        \setbox0=\hbox{%
                \advance\dimen6-.5\wd1
                \kern\dimen6
                $\vcenter to\proofdotnumber\proofdotseparation
                        {\leaders\box1\vfill}$%
                \unhbox\proofrulename}%
\else   \dimen6=\fontdimen22\the\textfont2 
        \dimen7=\dimen6
        \advance\dimen6by.5\proofrulebreadth
        \advance\dimen7by-.5\proofrulebreadth
        \setbox0=\hbox{%
                \kern\shortenproofleft
                \ifdoubleproof
                \then   \hbox to\dimen0{%
                        $\mathsurround0pt\mathord=\mkern-6mu%
                        \cleaders\hbox{$\mkern-2mu=\mkern-2mu$}\hfill
                        \mkern-6mu\mathord=$}%
                \else   \vrule height\dimen6 depth-\dimen7 width\dimen0
                \fi
                \unhbox\proofrulename}%
        \ht0=\dimen6 \dp0=-\dimen7
\fi
\let\doll\relax
\ifwasinsideprooftree
\then   \let\VBOX\vbox
\else   \ifmmode\else$\let\doll=$\fi
        \let\VBOX\vcenter
\fi
\VBOX   {\baselineskip\proofrulebaseline \lineskip.2ex
        \expandafter\lineskiplimit\ifproofdots0ex\else-0.6ex\fi
        \hbox   spread\dimen5   {\hfi\unhbox\proofabove\hfi}%
        \hbox{\box0}%
        \hbox   {\kern\dimen2 \box\proofbelow}}\doll%
\global\dimen2=\dimen2
\global\dimen3=\dimen3
\egroup 
\ifonleftofproofrule
\then   \shortenproofleft=\dimen2
\fi
\shortenproofright=\dimen3
\onleftofproofrulefalse
\ifinsideprooftree
\then   \hskip.5em plus 1fil \penalty2
\fi
}

\begin{document}

		\title{Intersection Types for a Computational lambda-Calculus with Global Store}

\address{riccardo.treglia@gmail.com}

\author{Ugo de'Liguoro \\
	Dipartimento di Informatica,
	Universit\`a di Torino \\ C.so Svizzera 185, Torino, Italy\\
	ugo.deliguoro@unito.it
	\and Riccardo Treglia\\
	Universitas Mercatorum\\
	riccardo.treglia@gmail.com } 
		
\maketitle

\runninghead{U. de'Liguoro, R. Treglia}{Intersection Types for a Computational lambda-Calculus with Global Store}		
		
\begin{abstract}
We study the semantics of an untyped lambda-calculus equipped with operators representing read and write operations from and to a global store. We adopt the monadic approach to model side-effects and treat read and write as algebraic operations over a monad. We introduce operational and denotational semantics and a type assignment system of intersection types and
prove that types are invariant under the reduction and expansion
of term and state configurations. Finally, we characterize convergent terms via their typings, establishing the adequacy of the denotational semantics
w.r.t. the operational semantics. 

\end{abstract}

%
%



		

\section{Introduction}
\label{sec:Introduction}

The problem of integrating non-functional aspects into functional programming languages goes back to the early days of functional programming;
nowadays, even procedural and object-oriented languages embody more and more features from declarative languages, renewing interest and
motivations in the investigation of higher-order effectful computations.

Focusing on side-effects, the methods used in the early studies to treat the store, 
e.g.   \cite{FelleisenFriedman89,Tofte90,WrightFelleisen94} and \cite{SwarupReddy91}, are essentially additive:
update and dereferentiation  primitives are added to a (often typed) $\lambda$-calculus, possibly with constructs to dynamically create and initialize
new locations. 
Then, a posteriori, tools to reason about such calculi are built either by extending the type discipline or by means of denotational and operational semantics, or both. Although our main concern here is the study of an imperative $\lambda$-calculus, we are not following the same path of previous treatments
of this topic in the literature, rather we insist on the construction of the type assignment system from first principles and denotational semantics of effectuful calculi in general, as in the companion paper \cite{deLiguoroT23}; differently than in our previous work we now consider the operational semantics of the imperative calculus.

A turning point in the study of computational effects has been Moggi's idea to model effectful computations as morphisms of the Kleisli category of a monad in \cite{Moggi'91}, starting a very rich thread in investigating programming language foundations based on categorical semantics which is still flourishing.
The methodological advantage is that we have a uniform and abstract
way of speaking of various kinds of effects, and the definition of equational logics to reason about them. At the same time computational effects 
including side effects can be safely embodied into purely functional languages without disrupting their declarative nature, as shown by 
Wadler's \cite{Wadler92,Wadler-Monads} together with a long series of papers
by the same author and others.

Initiating with \cite{deLiguoroTreglia20}, we have approached the issue of modelling effects in a computational 
$\lambda$-calculus using domain logic and intersection types. Since \cite{BCD'83}, intersection types have been understood
as a type theoretic tool to construct $\lambda$-models, which has systematically related to domain theory in \cite{Abramsky'91}.
As such, intersection types can serve as an intermediate layer between the abstract denotational semantics and 
the concrete syntax of the calculus, which we want to extend to the case of Moggi's computational calculi.

The idea is to model untyped computational $\lambda$-calculi into  domains satisfying the equation $D = D \to TD$, clearly
reminiscent of Scott's $D = D \to D$ reflexive object, where $T$ is a monad (see \cite{Moggi'88}, \S 5 and Section \ref{sec:denotational} below). Since intersection
types denote compact points in such a domain, we can recover from the domain theoretic definition of the monad 
all the information needed to build a sound and complete type assignment system.

The study in \cite{deLiguoroTreglia20} was confined to a generic monad represented by an uninterpreted symbol $T$ in the type syntax, 
without considering algebraic operators. In the
present paper, we address the issue of modelling an imperative $\lambda$-calculus equipped with a global store, and
read and write operations. In this case, the monad $T$ is instantiated to a variant of the state monad in \cite{Moggi'91}, called the
partiality and state monad in \cite{LagoGL17}, which we denote by $\GS$. To the calculus syntax, we add denumerably many
operations $\textit{get}_{\ell}$ and $\textit{set}_{\ell}$, indexed over an infinite set of locations, and define an operational semantics in SOS style, which turns out to be the small-step correspondent to the big-step operational semantics
proposed in \cite{amsdottorato9075}, chap. 3. We call $\lamImp$ the resulting calculus.

By solving the domain equation $D = D \to \GS D$ we get a model of $\lamImp$, whose theory includes
the convertibility relation induced by the reduction and the monadic laws. Such a model, when
constructed in the category of $\omega$-algebraic lattices provides the basis to define an intersection
type system, which is an instance of the generic one in  \cite{deLiguoroTreglia20}. Without delving
into the details of the process for obtaining such a system from the semantics, which has been treated in 
\cite{deLiguoroT23}, we provide an interpretation of types in the model, and prove type invariance under subject reduction and expansion, culminating with the characterization theorem of convergent terms via their typings in the system.

\paragraph{Contents}
After introducing the syntax of $\lamImp$ in Section \ref{sec:imp-calculus}, we define its operational semantics in Section \ref{sec:imp-operational} and the convergence relation in  
Section \ref{sec:imp-convergence}. Denotational semantics is treated in Section \ref{sec:denotational} to provide the mathematical background and motivations for constructing the intersection type theories and the type system in  Section \ref{sec:imp-intersection}. In  Section \ref{sec:type-interpretation} we show the soundness of the natural interpretation of typing judgments w.r.t. the denotational model; then, in 
Section \ref{sec:imp-invariance} we prove the type invariance property w.r.t. the operational semantics; eventually, 
in Section \ref{sec:imp-characterization} we establish
the main result of the paper, namely the characterization of convergent terms via their typing.
Section \ref{sec:imp-Related} is devoted to related papers and to the discussion of some issues. Finally, we conclude.

\medskip
We assume familiarity with $\lambda$-calculus and intersection types; a comprehensive reference is \cite{BarendregtDS2013} Part III. 
Concerning domain theory and the relation of intersection types to the category of $\omega$-algebraic lattices useful references are 
\cite{AbramskyJung94, Amadio-Curien'98}. Furthermore,
some basic concepts from category theory are used, for which the introductory text \cite{PierceCatTheory-91} should suffice.
We have kept the paper self-contained as much as possible, but some previous knowledge about computational monads and algebraic effects is recommended.
The interested reader might consult \cite{BentonHM00} and \cite{amsdottorato9075} chap. 3. 

\medskip
This paper is a largely revised and extended version of the conference paper \cite{deLiguoroT21a}.
 


\section{An untyped imperative $\lambda$-calculus}
\label{sec:imp-calculus}

Imperative extensions of the $\lambda$-calculus, both typed and type-free, are usually
based on the call-by-value $\lambda$-calculus, enriched with constructs for reading and writing to
the store. Aiming at exploring the semantics of side effects in computational calculi, where 
``impure'' functions are modelled by pure ones sending values to computations in the sense of 
\cite{Moggi'91}, we consider the computational core $\lc$ from \cite{deLiguoroTreglia20}, to which
we add syntax denoting algebraic effect operations \`{a} la Plotkin and Power \cite{PlotkinP02,PlotkinP03,Power06}
over a suitable state monad. 

Let $\Label = \Set{\ell_0, \ell_1, \ldots}$ be a denumerable set of abstract {\em locations}. Borrowing notation from 
\cite{amsdottorato9075}, we consider denumerably many operator symbols
 $\textit{get}_{\ell}$ and $\textit{set}_{\ell}$, obtaining:

\begin{definition}[Term syntax]\label{def:syntax}
\[\begin{array}{rrcll}
\ValTerm ~ \ni & V,W & ::= & x \mid \lambda x.M \\ 
\ComTerm ~ \ni & M,N & ::= & \unit{V} \mid M \Bind V  \\
                   &        & \mid & \get{\ell}{\lambda x.M} \mid \set{\ell}{V}{M} & (\ell \in \Label) \\ 
\end{array}\]
\end{definition}

\noindent
As for $\lc$, terms are of either sorts $\ValTerm$ and $\ComTerm$, representing values and computations, respectively.
The new constructs are $\get{\ell}{\lambda x.M} $ and $\set{\ell}{V}{M}$.
The variable $x$ is bound in $\lambda x.M$ and $\get{\ell}{\lambda x.M}$; 
terms are identified up to renaming of bound variables so that the capture avoiding substitution $M\Subst{V}{x}$ is always well defined;
$FV(M)$ denotes the set of free variables in $M$.
We call $\ValTerm^{\,0}$ the subset of closed $V \in \ValTerm$; similarly for $\ComTerm^0$.

With respect to the syntax of the ``imperative $\lambda$-calculus'' in \cite{amsdottorato9075}, 
we do not have the {\it let} construct, nor the application $VW$ among values. The justification is the same as for the computational core, indeed.
These constructs are definable: 
\[\Let{x}{M}{N} \;\equiv\; M \Bind (\lambda x.N) \qquad VW \;\equiv\; \unit{W} \Bind V\] 
where $\equiv$ is syntactic identity.
In general, application among computations can be encoded by $MN \equiv M \Bind (\lambda z. N \Bind z)$, where $z$ is  fresh.

In a sugared notation from functional programming languages, we could have written:
\[ \Let{x}{\; !\ell}{M} ~\equiv~ \get{\ell}{\lambda x.M}  \qquad 
 \ell := V; M ~\equiv~ \set{\ell}{V}{M} \]
representing location dereferentiation and assignment. Observe that we do not consider locations as
values; consequently they cannot be dynamically created like with the \textbf{ref} operator from ML,  nor
is it possible to model aliasing. On the other hand, since the calculus is untyped, ``strong updates'' are allowed. In fact,  when evaluating $\set{\ell}{V}{M}$,
the value $V$ to which the location $\ell$ is updated bears no relation
to the previous value, say $W$, of $\ell$ in the store: indeed, in our type assignment system $W$ and $V$ may well have completely different types. This will be, of course, a major challenge when designing the type assignment
 system in the next sections.


\section{Operational semantics}
\label{sec:imp-operational}

\newcommand{\slist}{\textit{sl}}
\newcommand{\nf}{\textrm{\rm nf}}
\newcommand{\fun}[1]{\overline{#1}}

\newcommand{\bfS}{\textbf{S}}
\newcommand{\bfV}{\textbf{V}}

\newcommand{\upd}[3]{\Aset{#1}(#2, #3)}
\newcommand{\lkp}[2]{\Aget{#1}(#2)}

\newcommand{\LkpSort}{\textit{Lkp}}

We define the operational semantics of our calculus 
via a reduction relation $\SmallStep{M}{s}{N}{t}$, where $M,N \in \ComTerm$ and $s,t$
are {\em store} terms, which are defined below.

\begin{definition}[Store and Lookup terms]\label{def:store-terms}
Let $V$ vary over $\ValTerm$; then define:

\[\begin{array}{rcll}
\StoreSort \ni s & ::= & \emp \mid \upd{\ell}{u}{s} \\ [1mm]
\LkpSort \ni u & ::= & V \mid \lkp{\ell}{s} & \ell \in \dom{s} \\
~ \\
\dom{\emp} & = & \emptyset \\ [1mm]
\dom{\upd{\ell}{u}{s}} & = & \Set{\ell} \cup \dom{s}
\end{array}\]
\end{definition}

Store terms represent finite mappings from $\Label$ to $\ValTerm$.
$\emp$ is the {\em empty store}, that is the everywhere undefined map; 
$\upd{\ell}{V}{s}$ is the {\em update} of $s$, representing the same map as $s$, 
but for $\ell$ where it holds $V$. 

To compute the value of the store $s$ at location $\ell$ we add the expressions 
$\lkp{\ell}{s} \in \LkpSort$, whose intended meaning  is the {\em lookup} (partial) function searching 
the value of $\ell$ in the store $s$; more precisely $\lkp{\ell}{s}$ picks the value $V$ from the leftmost
outermost occurrence in $s$ of a subterm of the shape $\upd{\ell}{V}{s'}$, if any.

To avoid dealing with undefined expressions like $\lkp{\ell}{\emp}$, 
we have asked that $\ell \in \dom{s}$ for the lookup expression $\lkp{\ell}{s}$ to be well-formed. Then
the meaning of store and lookup terms can be defined axiomatically as follows. 

\begin{definition}[Store and Lookup Axioms]\label{def:store-terms-axioms}
The axioms of the algebra of store terms and well-formed lookup expressions are the following:
\begin{enumerate}
\item \label{def:store-terms-axioms-a} $\lkp{\ell}{\upd{\ell}{u}{s}} = u$
\item \label{def:store-terms-axioms-b} $\lkp{\ell}{\upd{\ell'}{u}{s}} = \lkp{\ell}{s}$ if $\ell \neq \ell'$
\item \label{def:store-terms-axioms-c} $\upd{\ell}{\lkp{\ell}{s}}{s} = s$
\item \label{def:store-terms-axioms-d} $\Aset{\ell}(U, \Aset{\ell}(W, s)) = \Aset{\ell}(U, s)$
\item \label{def:store-terms-axioms-e} $\Aset{\ell}(U, \Aset{\ell'}(W, s)) =\Aset{\ell'}(W, \Aset{\ell}(U, s))$ if $\ell \neq \ell'$
\end{enumerate}

\end{definition}

We write $\der s = t$ to stress that the equality $s = t$ has been derived from the above axioms using reflexivity,
symmetry, transitivity, and congruence of equational logic.

The equalities in Definition \ref{def:store-terms-axioms} are folklore for terms representing the store in the literature:
see e.g. \cite{Mitchell'96}, chap. 6; also they are essentially, albeit not literally, 
the same as those ones for global state in \cite{PlotkinP02}. This is a non-trivial and decidable theory;
since we have not found any good reference to establish the properties we shall use in the subsequent sections, 
we devote the next paragraphs to the study of this theory. 
The main results are Theorem \ref{thr:store-axioms-completeness} and Corollary \ref{cor:s-nf-decidable}.

\begin{lemma}\label{lem:lkp(s)=V}
$ \ell \in \dom{s} \Then \exists V \in \ValTerm.\; \lkp{\ell}{s} = V$
\end{lemma}

\begin{proof}
By induction over $s$; since $\dom{s} \neq \emptyset$ we have that $s \not\equiv \emp$ and the lookup expression $\lkp{\ell}{s}$ is well-formed,
so that there are two cases to consider:
\begin{description}
\item $s \equiv \upd{\ell}{u}{s'}$: then by axiom \ref{def:store-terms-axioms}.\ref{def:store-terms-axioms-a} we have 
	$\lkp{\ell}{\upd{\ell}{u}{s'}} = u$. If $u \equiv V \in \ValTerm$ then we are done; otherwise $u \equiv \lkp{\ell}{s''}$, and
	the size of $s''$ is strictly smaller than that of $s$, so that the thesis follows by induction.
	
\item $s \equiv \upd{\ell'}{u}{s'}$, with $\ell' \neq \ell$: then $\lkp{\ell}{\upd{\ell'}{u}{s'}} = \lkp{\ell}{s'}$ 
	by axiom \ref{def:store-terms-axioms}.\ref{def:store-terms-axioms-b}, and the thesis follows by induction.
\end{description}
\vspace{-2ex}
\end{proof}

As expected, the relation $\der \lkp{\ell}{s} = V$ is functional:

\begin{lemma}\label{lem:store-unicity}
If $\der \lkp{\ell}{s} = V$ and $\der \lkp{\ell}{s} = W$ then $V \equiv W$.
\end{lemma}
\vspace{-2ex}

\begin{proof}
By induction over the derivations of $\lkp{\ell}{s} = V$.
\end{proof}

\begin{definition}\label{def:extensional-store-terms}
We say that $s, t \in \StoreSort$ are {\em extensionally equivalent}, written $s \simeq t$, if there exists $L \subseteq \Label$ such that:
\begin{enumerate}
\item $\dom{s} = L = \dom{t}$
\item $\forall \ell \in L.\; \lkp{\ell}{s} = \lkp{\ell}{t}$
\end{enumerate}
\end{definition}

\begin{lemma}\label{lem:store-soundness}
$\der s = t \Then s \simeq t$
\end{lemma}
\begin{proof}
The thesis holds immediately of the axioms \ref{def:store-terms-axioms}.\ref{def:store-terms-axioms-c}-\ref{def:store-terms-axioms}.\ref{def:store-terms-axioms-d};
then the proof is a straightforward induction over the derivation of $s = t$.
\end{proof}

Toward establishing the converse of Lemma \ref{lem:store-soundness} let us define $s\setminus \ell \in \StoreSort$ by:
\[\begin{array}{rlll}
\emp \setminus \ell & \equiv & \emp \\
\upd{\ell}{u}{s} \setminus \ell & \equiv & s\setminus \ell \\
\upd{\ell'}{u}{s} \setminus \ell & \equiv &  \upd{\ell'}{u}{s \setminus \ell} & \mbox{if $\ell' \neq \ell$} \\
\end{array}\]

\begin{lemma}\label{lem:s-without-ell}
\hfill
\begin{enumerate}
\item\label{lem:s-without-ell-a} $s \setminus \ell \equiv s \Iff \ell \not\in \dom{s}$
\item\label{lem:s-without-ell-b} $\dom{s \setminus \ell} = \dom{s} \setminus \Set{\ell}$
\item\label{lem:s-without-ell-c} $s \simeq t \Then (s\setminus \ell) \simeq (t \setminus \ell)$
\end{enumerate}
\end{lemma}
\begin{proof}
Parts (\ref{lem:s-without-ell-a}) and (\ref{lem:s-without-ell-b})  are immediate consequences of the definitions; part (\ref{lem:s-without-ell-c}) follows
from (\ref{lem:s-without-ell-a}) and (\ref{lem:s-without-ell-b}).
\end{proof}

\begin{lemma}\label{lem:s-whithout-ell-main}
$\ell \in \dom{s} \Then \; \der s = \upd{\ell}{\lkp{\ell}{s}}{s \setminus \ell}$
\end{lemma}

\begin{proof}
By induction over $s$. By the hypothesis $\ell \in \dom{s}$ we have that $\lkp{\ell}{s}$ is well-formed, and $s \not\equiv \emp$; so that we have the cases:
\begin{description}
\item $s \equiv \upd{\ell}{V}{s'}$: then we have $\der \lkp{\ell}{s} = V$ by axiom \ref{def:store-terms-axioms}.\ref{def:store-terms-axioms-a}
	so that 
	\[ \der \upd{\ell}{V}{s'} = \upd{\ell}{\lkp{\ell}{s}}{s'}  \]
	Now if $\ell \not \in \dom{s'}$ we have
	\[\begin{array}{rcll}
	 \upd{\ell}{\lkp{\ell}{s}}{s'} & \equiv & \upd{\ell}{\lkp{\ell}{s}}{s' \setminus \ell} &   \mbox{by  Lemma \ref{lem:s-without-ell}.\ref{lem:s-without-ell-a} } \\
	 & \equiv & \upd{\ell}{\lkp{\ell}{s}}{s \setminus \ell} & \mbox{by $s \setminus \ell \equiv \upd{\ell}{V}{s'} \setminus \ell \equiv s' \setminus \ell$}
	 \end{array}\]
	If instead $\ell \in \dom{s'}$ then $\lkp{\ell}{s'}$ is well-formed and
	\[\begin{array}{rcll}
	\upd{\ell}{\lkp{\ell}{s}}{s'}  & = & \upd{\ell}{\lkp{\ell}{s}}{\upd{\ell}{\lkp{\ell}{s'}}{s' \setminus \ell}}  & \mbox{by induction} \\
	& = & \upd{\ell}{\lkp{\ell}{s}}{s' \setminus \ell} & \mbox{by axiom \ref{def:store-terms-axioms}.\ref{def:store-terms-axioms-d}}
	\end{array}\]

\item $s \equiv \upd{\ell'}{V}{s'}$ and $\ell' \neq \ell$: in this case $\ell \in \dom{s} = \dom{ \upd{\ell'}{V}{s'}}$ implies
	$\ell \in \dom{s'}$, so that 
	\[\begin{array}{rcll}
	\upd{\ell'}{V}{s'} & = & \upd{\ell'}{V}{\upd{\ell}{\lkp{\ell}{s'}}{s' \setminus \ell}} & \mbox{by induction} \\
	& = & \upd{\ell}{\lkp{\ell}{s'}}{\upd{\ell'}{V}{s' \setminus \ell}} & \mbox{by axiom \ref{def:store-terms-axioms}.\ref{def:store-terms-axioms-e} } \\
	& \equiv & \upd{\ell}{\lkp{\ell}{s'}}{\upd{\ell'}{V}{s'}  \setminus \ell} & \mbox{by definition of $\upd{\ell'}{V}{s'} \setminus \ell$} \\
	& = & \upd{\ell}{\lkp{\ell}{s}}{\upd{\ell'}{V}{s'} \setminus \ell } &  \mbox{as $\der \lkp{\ell}{\upd{\ell'}{V}{s'}} = \lkp{\ell}{s'}$ } \\
	&&& \mbox{by axiom \ref{def:store-terms-axioms}.\ref{def:store-terms-axioms-b} }
	\end{array}\]

\end{description}
\end{proof}

Next, we define $\nf(s) \in \StoreSort$, a normal form of $s$, by:
\[\begin{array}{rlll}
\nf(\emp) & \equiv & \emp \\
\nf(\upd{\ell}{u}{s}) & \equiv & \upd{\ell}{u}{\nf(s \setminus \ell)}
\end{array}\]
In $\nf(s)$ each $\ell \in \dom{s}$ occurs just once.

\begin{corollary}\label{cor:store-terms-completeness}
$\der s = \nf(s)$
\end{corollary}

\begin{proof}
By simultaneous induction over $s$ and the cardinality of $\dom{s}$. If $\dom{s} = \emptyset$ then $s \equiv \emp \equiv \nf(\emp)$.
Otherwise, suppose $s \equiv \upd{\ell}{u}{s'}$, then 
\[\begin{array}{rcll}
\nf(\upd{\ell}{u}{s'}) & \equiv & \upd{\ell}{u}{\nf(s' \setminus \ell)} & \mbox{by definition of $\nf$} \\
\end{array}\]
Now if $\ell \not\in \dom{s'}$ then $s' \setminus \ell \equiv s'$ by Lemma \ref{lem:s-without-ell}.\ref{lem:s-without-ell-a}
and $\der s' = \nf(s')$ by induction, and we are done.

If $\ell \in \dom{s'}$ then, recalling that $s \equiv \upd{\ell}{u}{s'}$:
\[\begin{array}{rcll}
s & = & \upd{\ell}{\lkp{\ell}{s}}{s \setminus \ell} & \mbox{by Lemma \ref{lem:s-whithout-ell-main} } \\
& = & \upd{\ell}{\lkp{\ell}{s}}{\nf(s \setminus \ell)} & \mbox{by Lemma \ref{lem:s-without-ell}.\ref{lem:s-without-ell-b} and by ind. since $|s| > |\dom{s \setminus \ell}|$} \\
& \equiv & \upd{\ell}{\lkp{\ell}{s}}{\nf(s' \setminus \ell)} & \mbox{since $s \setminus \ell \equiv \upd{\ell}{u}{s'} \setminus \ell \equiv s'\setminus \ell$} \\
& = & \upd{\ell}{u}{\nf(s' \setminus \ell)} & \mbox{since $\lkp{\ell}{s} \equiv \lkp{\ell}{ \upd{\ell}{u}{s'}} = u$} \\
& \equiv & \nf(s) & \mbox{by definition of $\nf$}
\end{array}\]

\end{proof}

\begin{theorem}[Completeness of Store Axioms]\label{thr:store-axioms-completeness}
$\der s = t \Iff s \simeq t$
\end{theorem}

\begin{proof}
The only if part is proved by Lemma \ref{lem:store-soundness}. To prove the if part, assume $s \simeq t$ and 
let $L = \dom{s} = \dom{t}$, which is finite; then we reason by induction over the cardinality of $L$. 
If $L = \emptyset$ then $s \equiv \emp \equiv t$ and the thesis follows by reflexivity.
Otherwise, let $\ell \in L$ be arbitrary; by  Lemma \ref{lem:s-without-ell}.\ref{lem:s-without-ell-c} 
we have that $s \setminus \ell \simeq t \setminus \ell$, therefore
we may assume by induction 
$\der s \setminus \ell = t \setminus \ell$ since $L \setminus \Set{\ell}  = \dom{s \setminus \ell} = \dom{t \setminus \ell}$
has cardinality $|L| - 1$. Now
\[\begin{array}{rcll}
s & = & \upd{\ell}{\lkp{\ell}{s}}{s \setminus \ell} & \mbox{by Lemma  \ref{lem:s-without-ell}.\ref{lem:s-without-ell-c}} \\
& = & \upd{\ell}{\lkp{\ell}{t}}{s \setminus \ell} & \mbox{by the hypothesis $s \simeq t$} \\
& = & \upd{\ell}{\lkp{\ell}{t}}{t \setminus \ell} & \mbox{by induction} \\
& = & t & \mbox{by Lemma \ref{lem:s-without-ell}.\ref{lem:s-without-ell-c}} 
\end{array}\]

\end{proof}
The consequence of Theorem \ref{thr:store-axioms-completeness} is that each store term is equated to a normal form:

\begin{corollary}\label{cor:s-nf-decidable}
If $s \in \StoreSort$ with non empty $\dom{s} = \Set{\ell_1, \ldots , \ell_n}$ then
there exist $V_1, \dots , V_n \in \ValTerm$ such that
\[\der s = \upd{\ell_1}{V_1}{\cdots \upd{\ell_n}{V_n}{ \emp} \cdots}\]
Therefore the algebra of store and lookup terms is decidable.
\end{corollary}

\begin{proof}
Observe that, if $\dom{s} = \Set{\ell_1, \ldots , \ell_n}$ is non empty then
\[\nf(s) \equiv \upd{\ell_1}{u_1}{\cdots \upd{\ell_n}{u_n}{ \emp} \cdots}\] 
By Lemma \ref{lem:lkp(s)=V} we know that there exists $V_i \in \ValTerm$ such that
$\der u_i = V_i$ for all $i = 1, \ldots, n$, and these are unique for each $u_i$ by
 Lemma \ref{lem:store-unicity}. Then we have
\[\der \nf(s) = \upd{\ell_1}{V_1}{\cdots \upd{\ell_n}{V_n}{ \emp} \cdots}\] 
where the only differences between the left and the right-hand sides are in the 
ordering of the $\ell_i$, which is immaterial by axiom 
\ref{def:store-terms-axioms}.\ref{def:store-terms-axioms-e}.
By Corollary \ref{cor:store-terms-completeness} we conclude that
\[\der s = \upd{\ell_1}{V_1}{\cdots \upd{\ell_n}{V_n}{ \emp} \cdots}\]

Combining this with Theorem \ref{thr:store-axioms-completeness}, we conclude that
$\der s = t$ if and only if $\dom{s} = L = \dom{t}$ and $\nf(s) \simeq \nf(t)$, 
which is decidable as $\nf$ is computable and
extensional equivalence is decidable.

\end{proof}

\myparagraph{The reduction relation}
A {\em configuration} is a pair $(M,s)$, with $M \in \ComTerm$ and $s\in \StoreSort$; 
then the {\em one-step reduction} is the binary relation over configurations inductively defined by the rules in Figure~\ref{fig:reduction}.

The reduction relation is deterministic, reflecting the strictly sequential nature of evaluation for programming languages with
side-effects. A configuration of the shape $(\unit{V},t)$ is irreducible, and it is the result of
the evaluation of  $(M,s)$ whenever $\SmallStepStar{M}{s}{\unit{V}}{t}$, where $\RedStar$ is the reflexive and 
transitive closure of $\Red$. Infinite reductions exist; consider the term:
\[\Omega_c \equiv \unit{(\lambda x. \unit{x} \Bind x)} \Bind  (\lambda x. \unit{x} \Bind x)
\]
which is such that $\SmallStep{\Omega_c}{s}{\Omega_c}{s}$, for any $s$.

Not every irreducible configuration represents some properly terminating computation; the simplest example is
$(\get{\ell}{\lam x.\unit{x}}, \emp)$, because $\Aget{\ell}(\emp)$ is undefined (and even not well-formed). In general, 
the set of {\em blocked configurations} can be inductively defined by:
\[ \begin{array}{lcll}
(B,s) & ::= & (\get{\ell}{\lambda x.M}, s) & \mbox{for}~ \ell \not \in \dom{s} \\
& \mid & (B \Bind V, s) 
\end{array}\]

\begin{figure}
\begin{equation*}
\begin{array}{c@{\hspace{2mm}}l}
\SmallStep{ \unit{V} \Bind (\lambda x.M) } {s} { M\Subst{V}{x} } {s} & \betaR\\ [2mm]

\prooftree
	\SmallStep{M}{s}{N}{t}
\justifies
	\SmallStep {M \Bind V} {s} { N \Bind V} {t}
\endprooftree
& \BindR
\\ [6mm]
\prooftree
	\Aget{\ell} (s) = V
\justifies
	\SmallStep { \get{\ell}{\lambda x.M} } {s} { M \Subst{V}{x} } {s}
\endprooftree
& 	\getredR
\\ [6mm]
\SmallStep {\set{\ell}{V}{M}}{s}{M}{\Aset{\ell}(V, s)} & \setredR
\end{array}
\end{equation*}

\caption{Reduction relation.}\label{fig:reduction}
\end{figure}

\begin{example}\label{ex:overriding}
To see how the mutation of the value associated to some location $\ell$ is modelled in the operational semantics,
consider the following reduction, where we omit external brackets of configurations for readability:
\[\begin{array}{crl}
 & \set{\ell}{W}{ \set{\ell}{V}{ \get{\ell}{\lambda x.\unit{x} }  } },& s \\
 \Red &  \set{\ell}{V}{ \get{\ell}{\lambda x.\unit{x} }  } ,& \Aset{\ell}(W, s) \\
 \Red &  \get{\ell}{\lambda x. \unit{x}},& \Aset{\ell}(V, \Aset{\ell}(W, s) ) \\
 \Red & \unit{V},& \Aset{\ell}(V, \Aset{\ell}(W, s) )  
\end{array}\]
where the last step is justified by the equation 
\[\Aget{\ell}(\Aset{\ell}(V, \Aset{\ell}(W, s) ) ) = V\]
Then we say that the value $W$,
formerly associated to $\ell$, has been {\em overridden} by $V$ in $\Aset{\ell}(V, \Aset{\ell}(W, s) ) $.
Indeed $\der \Aset{\ell}(V, \Aset{\ell}(W, s) ) =  \Aset{\ell}(V, s)$.
\end{example}

\begin{example}\label{ex:reduction}
Define the abbreviation:
$M ; N \equiv M \Bind \mutelambda . N$ for {\em sequential composition}, 
where $\underbar{~}$ is a dummy variable not occurring in $N$;
then, omitting external brackets of configurations as before:
\[\begin{array}{lcrl}
& & \set{\ell}{V}{\unit{W}} \, ; \, \get{\ell}{\lambda x.N} ,& s \\ [1mm]
& \equiv & \set{\ell}{V}{\unit{W}} \Bind \mutelambda .  \get{\ell}{\lambda x.N} ,& s \\ [1mm]
& \Red & \unit{W}  \Bind \mutelambda .  \get{\ell}{\lambda x.N} ,& \Aset{\ell}(V, s) \\ [1mm]
& \Red & \get{\ell}{\lambda x.N} ,& \Aset{\ell}(V, s) \\ [1mm]
& \Red & N\Subst{V}{x},& \Aset{\ell}(V, s)
\end{array}\]
Notice that, while the value $W$ is discarded according to the semantics of sequential composition, 
the side-effect of saving $V$ to the location $\ell$ binds $x$ to $V$ in $N$.
\end{example}


\section{Convergence}\label{sec:imp-convergence}

Following \cite{Pitts98operationalreasoning},
in \cite{amsdottorato9075} the operational semantics of the imperative $ \lambda $-calculus is defined via a {\em convergence predicate}.
This is a relation among configurations and their results, which in \cite{amsdottorato9075} are semantical objects. To adapt such a definition
to our syntactical setting, we use store terms instead.

Recall that $\ValTerm^{\,0}$ and $\ComTerm^0$ are the sets of closed values and computations, respectively. We say that $s \in \StoreSort$
is {\em closed} if $V \in  \ValTerm^{\,0}$ for all the value terms $V$ occurring in $s$; we denote the set of {\em closed store terms} by $\StoreSort^0$.
We say that the configuration $(M, s)$ is {\em closed} if both $M$ and $s$ are such; we call a {\em result} any pair $(V, s)$ of closed $V$ and $s$.

\begin{definition}[Big-step]\label{def:big-step}
The relation $\BigStep{M}{s}{V}{t}$ among the closed configuration $(M, s)$ and the result $(V, t)$
is inductively defined by the rules in Figure \ref{fig:convergence-pred}.
\end{definition}

\begin{figure}
	\centering
$\begin{array}{c@{\hspace{2mm}}l}
\prooftree
	\vspace{3mm}
\justifies
	\BigStep{\unit{V}}{s}{V}{s} 
	
\endprooftree
& \bigRVal
\\ [8mm]
\prooftree
	\BigStep{M}{s}{V}{s'}
	\quad
	\BigStep{N\Subst{V}{x}}{s'}{W}{t}
\justifies
	\BigStep{M \Bind (\lambda x.N)}{s}{W}{t}
\endprooftree
& \bigRBind
\\ [8mm]
\prooftree
	\Aget{\ell}(s) = V
	\quad
	\BigStep{ M\Subst{V}{x} }{s}{W}{t}
\justifies
	\BigStep{ \get{\ell}{\lambda x.M} }{s}{W}{t}
\endprooftree
& \bigRGet
\\ [8mm]
\prooftree
	\BigStep{ M }{\Aset{\ell}(V,s)}{W}{t}
\justifies
	\BigStep{ \set{\ell}{V}{M} }{s}{W}{t}
\endprooftree
& \bigRSet
\end{array}$
\caption{Convergence predicate.}\label{fig:convergence-pred}
\end{figure}

As suggested by the name used in \ref{def:big-step}, convergence is nothing else
than the big-step semantics corresponding to the small-step semantics we have defined via the reduction
relation in Figure \ref{fig:reduction}. 

\begin{lemma}\label{lem:big-small}
	 $ (M, s) \RedStar (N, t) \Then (M\Bind V, s) \RedStar (N\Bind V, t) $
\end{lemma}
\begin{proof}
By induction over the definition of $(M, s) \RedStar (N, t)$. The base case $(M,s) \equiv (N,t)$ is obvious. Otherwise $(M,s) \Red (M',s') \RedStar (N,t)$ for some
$M',s'$; then $(M \Bind V, s) \Red (M' \Bind V,s')$ by rule $\BindR$, and $(M' \Bind V,s') \RedStar (N \Bind V, t)$ by induction.

\end{proof}

\begin{lemma}\label{lem:small-big}
	$(M, s) \Red (N, s') \,\And\, (N, s') \Downarrow (V, t) \Then (M,s)\Downarrow( V, t)$
\end{lemma}

\begin{proof}
	By induction over 	$(M, s) \Red (N, s') $.
	\begin{description}
		\item Case ($\betac$):
		$\SmallStep{ M\equiv \unit{W} \Bind (\lambda x.M') } {s} 
		{  M'\Subst{W}{x} } {s}$; by hypothesis we know that $ ( M'\Subst{W}{y},s)\Downarrow (V,t) $,  then:
		\[\begin{array}{c}
			\prooftree
			\prooftree
			\vspace{3mm}
			\justifies
			(\unit{W},s)\Downarrow (W,s)
			\using \bigRVal
			\endprooftree
			 \qquad ( M'\Subst{W}{y},s)\Downarrow (V,t)
			\justifies
			(\unit{W} \Bind \lambda y.M',s)\Downarrow(V,t)
			\using \bigRBind
			\endprooftree
		\end{array}\]

		\item Case ($\Bind$-red): $ \SmallStep { M' \Bind W} {s} { N' \Bind W} {s'}  $ because $ \SmallStep{M'}{s}{N'}{s'} $.
		
		In this case we have $W\equiv \lam x.L$, since $W\in \ValTerm^{\,0}$. 
		By hypothesis $(N'\Bind \lam x.L ,s')\Downarrow (V, t)$, so there exist $W', s''$ such that 
		
		\[\prooftree
		\BigStep{N'}{s'}{W'}{s''}
		\quad
		\BigStep{L\Subst{W'}{x}}{s''}{V}{t}
		\justifies
		\BigStep{ M' \Bind \lam x.L}{s'}{V}{t}
		\using \bigRBind
		\endprooftree\]
		
		Since  $ \SmallStep{M'}{s}{N'}{s'} $ and $ \BigStep{N'}{s'}{W'}{s''} $, by induction hypothesis $ \BigStep{M'}{s'}{W'}{s''} $ and therefore we have the derivation
		
		\[\prooftree
		\BigStep{M'}{s'}{W'}{s''}
		\quad
		\BigStep{L\Subst{W'}{x}}{s''}{V}{t}
		\justifies
		\BigStep{ M' \Bind \lam x.L}{s'}{V}{t}
		\using \bigRBind
		\endprooftree\]
				
	\end{description}
	The remaining cases
		$\SmallStep { \get{\ell}{\lambda x.M'} } {s} { M' \Subst{W}{x} } {s}$ where $\Aget{\ell} (s) = W$, and 	
		$\SmallStep {\set{\ell}{W}{M'}}{s}{ M'}{\Aset{\ell}(W, s)\equiv s'}$, easily follow by induction.

\end{proof}

\begin{restatable}{proposition}{bigSmall}\label{prop:big-small}
For all $M \in \ComTerm^0$, $V \in \ValTerm^{\,0}$, and $s,t \in \StoreSort^0,$ we have:
\[ \BigStep{M}{s}{V}{t} \iff (M, s) \RedStar (\unit{V}, t) \]
\end{restatable}

\begin{proof}
The only if part is proved by induction over the definition of $\BigStep{M}{s}{V}{t}$, the case $\bigRVal$ is trivial. 
Consider the case $\bigRBind$. We know by induction that
\begin{description}
\item[$ HI_1 $] $ (M,S) \RedStar (\unit{V}, s')$
\item[$ HI_2 $] $ (N\Subst{V}{x}, s') \RedStar (\unit{W}, t) $
\end{description}
\[\begin{array}{rcll}
(M \Bind \lam x. N,s) &\RedStar& (\unit{V}\Bind \lam x.N, s')  & \text{ by } (HI_1) \text{ and Lemma \ref{lem:big-small}} \\
&\Red & (N\Subst{V}{x}, s') & \text{ by }\betac\\
&\RedStar& (\unit{W}, t) &\text{ by } (HI_2)
\end{array}\]

The cases of $\bigRGet$ and $\bigRSet$ are immediate by the induction hypothesis.

The if part is proved by induction over
$(M, s) \RedStar (\unit{V}, t)$.
The base case is $(M\equiv \unit{V}, s\equiv t)$ then $ \BigStep{\unit{V}}{s}{V}{s}  $ by $ \bigRVal $.
Otherwise, there exists $(N,s')$ such that 
$(M, s) \Red (N,s') \RedStar (\unit{V}, t)$. By induction $ \BigStep{N}{s'}{V}{t}  $ so that $ \BigStep{M}{s}{V}{t}  $  by Lemma \ref{lem:small-big}.

\end{proof}

Proposition \ref{prop:big-small} justifies the name ``result'' we have given to the pairs $(V,t)$; indeed, since the reduction relation is deterministic, 
the same holds for convergence so that  computations can be seen as partial functions from stores to results:
\begin{equation}\label{eq:opsem-func}
M(s) = \left \{ 
\begin{array}{ll}
	(V,t) & \mbox{if $\BigStep{M}{s}{V}{t}$} \\
	\mbox{undefined} & \mbox{else}
\end{array} \right.
\end{equation}
Notice that $M(s)$ is undefined if either the reduction out of $(M, s)$ is infinite, or if it reaches some blocked configuration
$(B, t)$.

The convergence predicate essentially involves the stores, that are part of configurations and results, and are dynamic entities,
much as when executing imperative programs. However, when reasoning about programs, namely (closed) computations, 
we abstract from the infinitely many stores that can be fed together with inputs to the program, and returned together
with their outputs. This motivates the following definition:

\begin{definition}[Convergence]\label{def:convergence}
For $M \in \ComTerm^0$ and $s \in \StoreSort^0$ we set:
\begin{enumerate}
\item $(M, s) \Downarrow ~ \iff ~ \exists V, t. ~ \BigStep{M}{s}{V}{t} $
\item $M \Downarrow ~ \iff ~  \forall s \in \StoreSort^0.\, (M, s) \Downarrow $
\end{enumerate}
\end{definition}

The definition of $(M, s) \Downarrow$ is similar to that in \cite{Pitts98operationalreasoning}, but simpler since we do not treat local stores. 
The definition of $M \Downarrow$ is equivalent to $(M, \emp) \Downarrow$. To see why,
consider the term $P\equiv \get{\ell}{\lam x.\unit{x}}$: the configuration $(P,s)$ converges if and only if $\Aget{\ell}(s)$ is defined;
in particular, the configuration $(P,\emp)$ is blocked and it does not yield any result.
Conversely, if a given configuration $(M, \emp)$ converges, either no \textit{get} operation occurs as the main operator 
of a configuration $(N,s')$
in the convergent  out of $(M, \emp)$, 
or if any such operation does, with say $N\equiv \get{\ell}{\lam x.M'}$, then there exists a term of shape $\set{\ell}{V}{N'}$ in a configuration
that precedes $(N,s')$ in the reduction path,
``initializing'' to $V$ the value of $\ell$. Therefore, if $(M,\emp)\Downarrow$ then $(M,s)\Downarrow$, for any $s$.

We say that $(M,s)$ is {\em divergent}, written $(M,s)\Uparrow$, if not $(M,s)\Downarrow$. By Proposition \ref{prop:big-small} a 
divergent configuration $(M,s)$ either reduces to a blocked configuration or the unique 
reduction out of $(M,s)$ is infinite.


\newcommand{\DinftyS}{D_\infty^\GS}

\section{Denotational semantics}\label{sec:denotational}

In this section, we construct a model $\DinftyS$ to interpret terms of the $\lamImp$ calculus. We first define precisely the {\em partiality and state monad} $\GS$
that is essential to model stores; then we recall Plotkin and Power's {\em generalised algebraic operators} over a monad, by adapting their definition 
the to the present setting. By solving the equation $D = D \to \GS\,D$ in a suitable category of domains $\CatDom$, which we choose as the category of $\omega$-algebraic lattices (see the references below), we get a domain $\DinftyS$ 
such that we can interpret values and computations in $\DinftyS$ and
$\GS\, \DinftyS$, respectively. This will be of use in Section \ref{sec:type-interpretation} to show the soundness of the type system. A more detailed treatment of the denotational semantics and of the filter-model construction for $\lamImp$ can be found in \cite{deLiguoroT23}. 

Recall that a concrete category $\CatC$ is essentially a set-like category; more precisely it is equipped with a faithful functor to the category of sets. If it is a cartesian closed category, shortly a ccc (see e.g. \cite[chap. 3]{PierceCatTheory-91}
and \cite[sec. 4]{AbramskyJung94}), then products can be seen as sets of pairs and exponents as sets of functions.

\begin{definition}\label{def:Cmonad}
{\bf (Computational Monad \cite{Wadler-Monads} \S 3)}
	Let $\CatC$ be a concrete cartesian closed category. A {\em functional computational monad}, henceforth just a {\em monad}
	over $\CatC$ is a triple $(T, \UnitSub \,, \Bind)$ where
	$T$ is a map over the objects of $\CatC$, and $\UnitSub \,$ and $\Bind$ are families of morphisms
	\[\begin{array}{c@{\hspace{2cm}}c}
		\UnitSub{A}: A \to TA & \Bind_{A,B}: TA \times (TB)^A  \to TB
	\end{array}\]
	such that, writing $\Bind_{A,B}$ as an infix operator and omitting subscripts:
	\[\begin{array}{l@{\hspace{0.5cm}}rcl@{\hspace{0.5cm}}l}
		\LeftUnit: & \UnitSub \, \, a \Bind f & = & f\,a \\
		\RightUnit: & m \Bind \UnitSub \, & = & m \\
		\Assoc: & (m \Bind f) \Bind g & = & m \Bind \metalambda d. (f\,d \Bind g)
	\end{array}\]
	where
$\metalambda d. (f\,d \Bind g)$ is the lambda notation for $d \mapsto f\,d \Bind g$.
\end{definition}

In case of a concrete ccc, Definition \ref{def:Cmonad} is equivalent to the notion of {\em Kleisli triple} $(T, \eta, \_^\dagger)$ (see e.g. \cite[Def.~1.2]{Moggi'91}).
The equivalence is obtained by observing that $\UnitSub{X} = \eta_X$ and $a \Bind f = f^\dagger (a)$. The second equation establishes the connection between the $\Bind$ operator with the map $\_^\dagger$, also called {\em extension operator}: if $f:X \to TY$ then $f^\dagger:TX \to TY$ is the unique map such that $f = f^\dagger \circ \eta_X$.

\medskip
The category $\CatDom$ of $\omega$-algebraic lattices
is a concrete ccc: objects are complete lattices, henceforth simply referred to as {\em domains}, 
and morphisms are Scott-continuous functions, namely preserving directed sups: 
see e.g. \cite[sec. 2 and 4]{AbramskyJung94} and \cite[chap. 1]{Amadio-Curien'98}. 
We now define the {\em partiality and state monad} as a monad over $\CatDom$.

\begin{definition}\label{def:statemon}
{\bf (Partiality and state monad)}
	Given the domain $S = (X_\bot)^\Label$ representing a notion of {\em state} with values in $X$, we define the partiality and state monad $(\GS, \UnitSub \,, \Bind)$, as the mapping
	\[\GS\,X = [S \to (X \times S)_\bot]\]
	where $(X \times S)_\bot$ is the lifting of the cartesian product $X \times S$, equipped with two (families of) operators $\UnitSub \,$ and $\Bind$ defined as follows:
	
	\[\begin{array}{l@{\hspace{0.5cm}}rcl@{\hspace{0.5cm}}l}
		\UnitSub \,\,x::= \metalambda \finMap. (x,\finMap)\\
		(c\Bind f)\finMap =	f^\dagger(c)(\finMap) ::= \left \{
		\begin{array}{lll}
			f \,(x) (\finMap')\, & \mbox{if $c(\finMap)= (x, \finMap')\neq \bot$} \\
			\bot                & \mbox{if $c(\finMap) =  \bot$}
		\end{array} \right.
	\end{array}\]
where we omit subscripts.
\end{definition}

In the definition of the monad $\GS$ there are two occurrences of the lifting functor $(\cdot)_\bot$ that incidentally is a monad itself. The first one is in the definition of $S = (X_\bot)^\Label$, where it is necessary to model the stores over $X$ as partial functions from $\Label$ to $X$; see
Definition \ref{def:store-lkp-config-interpretation} further on, where the semantics of store and lookup terms is formally stated. Here we use the set-theoretic notation $(\cdot)^\Label$ for (total) maps with domain $\Label$ for the latter is just a set, while the ordering over $S$
is inherited from that of $X_\bot$, assuming that $X$ is a domain.
The second place is $(X \times S)_\bot$ that again is motivated by the need to model computations as partial functions from $S$ to $X \times S$, as it will be apparent in Definition 
\ref{def:lambda-imp-model}.

In the proof of the next theorem we use standard concepts and techniques from domain theory and the solution of domain equations,
for which we refer the reader to \cite{AbramskyJung94} and \cite[chap. 7]{Amadio-Curien'98}.

\begin{theorem}\label{thm:domain-equation}
There exists a domain $\DinftyS$ such that the state monad $\GS$ with state domain 
$S = ((\DinftyS)_\bot)^\Label$ is a solution in $\CatDom$ to the domain equation:
\[D = [D \to \GS\,D]\]
Moreover, it is initial among all solutions to such an equation.
\end{theorem}

\begin{proof}
Consider the functor $FX = (X_\bot)^\Label$ and the mixed-variant bi-functor $G: \CatDom^\textit{op} \times \CatDom \to \CatDom$ by
\[G(X,Y) = [FX \to (Y \times FY)_\bot] \]
whose action on morphisms is illustrated by the diagram:

\[
\xymatrix@C=4em{
FX' \ar[r]^{Ff} \ar@{-->}[d]^{G(f,g)(\alpha)} 
  & FX \ar[d]_{\alpha} \\
(Y' \times FY')_{\bot} 
  & (Y \times FY)_{\bot} \ar[l]^{(g \times Fg)_{\bot}}
}
\]

where $f:X' \to X$, $g:Y \to Y'$ and $\alpha \in G(X,Y)$. Now it is routine to prove that
$G$ is locally continuous so that, by the inverse limit technique, 
we can find the initial solution $\DinftyS$ to the domain equation 
\begin{equation}\label{eq:domain-equation}
D \,\cong\,  [D \to G(D,D)]
\end{equation}
By setting $S = F\DinftyS = ((\DinftyS)_\bot)^\Label$ in the definition of $\GS$ we obtain $\GS\, \DinftyS  \,\cong\,  G(\DinftyS, \DinftyS)$ and hence
$\DinftyS  \,\cong\, [\DinftyS \to \GS\,\DinftyS]$
as desired.
\end{proof}

\myparagraph{Denotational semantics of terms}
Following ideas from \cite{PlotkinP02,PlotkinP03,Power06} and adopting notation from \cite{amsdottorato9075}, we say that in a concrete ccc a 
{\em generalized algebraic operator} over a monad $T$ is a morphism
\[\textbf{op}:  P \times (TX)^A \to TX \cong (TX)^A \to (TX)^P\]
where $P$ is the domain of parameters and $A$ of generalized arities, and the rightmost ``type'' is the interpretation in $\CatDom$ of Def. 1 in \cite{PlotkinP02}. 
Such operations must satisfy (see \cite{amsdottorato9075}, Def. 13):
\begin{equation}\label{eq:gen-op}
\textbf{op}(p,k) \Bind f = 
f^\dagger (\textbf{op}(p,k)) = \textbf{op}(p, f^\dagger \circ k)
= \textbf{op}(p, \metalambda x. (k(x) \Bind f))
\end{equation}
where $f: X \to TX$, $p \in P$ and $k:A \to TX$.

\medskip
We are now in place to interpret $\textit{set}_{\ell}$ and $\textit{get}_{\ell}$ as generalized algebraic operations over $\GS$. 
By taking $P = \textbf{1}$ and $A = D$ (writing $D = \DinftyS$ for short) we define:
\[\Sem{\textit{get}_{\ell}}:  \textbf{1} \times (\GS D)^D \to \GS D \simeq (\GS D)^D \to \GS D\quad
\mbox{ by } \quad  \Sem{\textit{get}_{\ell}} \, d \, \storeMap = d\, (\storeMap \,( \ell))\, \storeMap\] 
where $d \in D$ is identified with its image in $D \to \GS\,D$, and $\storeMap \in S = (D_\bot)^\Label$.
On the other hand, taking $P = D$ and $A = \textbf{1}$, we define:
\[\Sem{\textit{set}_{\ell}}: D \times (\GS D)^{\textbf{1}} \to \GS D \simeq D \times \GS D \to \GS D \quad
\mbox{ by }  \Sem{\textit{set}_{\ell}}(d, c)\, \storeMap = c (\storeMap[\ell \mapsto d])\] 
where $c \in \GS D = S \to (D \times S)_\bot$ and $\storeMap[\ell \mapsto d]$ is the store
sending $\ell$ to $d$ and it is equal to $\storeMap$, otherwise.

Then we interpret values from $\ValTerm$ in $D$
and computations from $\ComTerm$ in $\GS D$. More precisely we define the maps
$\SemD{\cdot}: \ValTerm \to \Env \to D$ and 
$\SemSD{\cdot}: \ComTerm \to \Env \to \GS D$,
where $\Env = \Var \to D$ is the set of environments interpreting term variables:

\begin{definition}\label{def:lambda-imp-model}
\label{def:lam-imp-mod}
A \textbf{$ \lamImp $-model} is a structure $\Dcat = (D, \GS, \SemD{\cdot},\SemSD{\cdot})$ such that:
\begin{enumerate}
	\item $D$ is a domain s.t. $D\cong D\to \GS D$ via $(\Phi,\Psi)$, where $\GS$ is the partiality and state monad of stores over $D$;
	\item for all $e \in \Env$, $V \in \ValTerm$ and $M \in \ComTerm$:
		\[\begin{array}{rcl}
		\Sem{x}^D e  & = &  e (x) \\ [1mm]
		\SemD{\lambda x.M} e  & = & \Psi( \metalambda d \in D. \, \SemSD{M} {e[x  \mapsto d]} ) \\ [1mm]
		\SemSD{\, \unit{V} \,} e & = & \UnitSub \, (\SemD{V} e) \\ [1mm]
		\SemSD{M \Bind V} e &  = & (\SemSD{M} e)\Bind \Phi (\SemD{V} e )\\ [1mm]
		\SemSD{ \get{\ell}{\lambda x.M}} e & =  & \Sem{\textit{get}_{\ell}} \, \Phi (\SemD{\lambda x.M} e)  \\ [1mm]
		\SemSD{ \set{\ell}{V}{M} } e  & = & \Sem{\textit{set}_{\ell}}(\SemD{V} e, \SemSD{M}  e) 
		\end{array}\]
\end{enumerate}
\end{definition}
By unravelling definitions and applying them to an arbitrary store $\storeMap \in S$, the last two clauses can be written:
\[\begin{array}{rcl}
\SemSD{ \get{\ell}{\lambda x.M}} e \, \storeMap & = & \SemSD{M} (e [x  \mapsto \storeMap(\ell)]) \, \storeMap \\ [1mm]
\SemSD{ \set{\ell}{V}{M} } e \, \storeMap & = & \SemSD{M}  e \, (\storeMap[\ell \mapsto \SemD{V} e])
\end{array}\]

We say that the equation $M = N$ is {\em true} in $\Dcat$, written $\Dcat \models M = N$, if $\SemSD{M} e = \SemSD{N} e$ for all $e \in \Env$.

\begin{proposition}\label{prop:true-eq}
The following equations are true in any $\lamImp$-model $\Dcat$:
\begin{enumerate}
\item \label{prop:true-eq-i} $\unit{V} \Bind (\lambda x.M) = M\Subst{V}{x}$
\item \label{prop:true-eq-ii} $M \Bind \lambda x. \unit{x} = M$
\item \label{prop:true-eq-iii} $(L \Bind \lambda x.M) \Bind \lambda y.N = L \Bind \lambda x. (M \Bind \lambda y. N)$
\item \label{prop:true-eq-iv} $\get{\ell}{\lambda x.M} \Bind W = \get{\ell}{\lambda x.(M \Bind W)}$
\item \label{prop:true-eq-v} $\set{\ell}{V}{M} \Bind W = \set{\ell}{V}{M \Bind W}$
\end{enumerate}
where $x \not \in FV(\lambda y.N)$  in (\ref{prop:true-eq-iii}) and $x \not\in FV(W)$ in (\ref{prop:true-eq-iv}).
\end{proposition}

\begin{proof}
By definition and straightforward calculations. For example, to see (\ref{prop:true-eq-iv}), let $\storeMap \in S$ be arbitrary
and $e' = e [x  \mapsto \storeMap(\ell)]$; then, omitting the apices of the interpretation mappings $\Sem{\cdot}$:
\[
\Sem{\get{\ell}{\lambda x.M} \Bind W} e \, \storeMap =
	\GSext{(\Sem{W} e)}(\Sem{M} e') \, \storeMap
\]
On the other hand:
\[
\Sem{\get{\ell}{\lambda x.(M \Bind W)}} e \, \storeMap 
= (\Sem{\lambda x.(M \Bind W)} e) \, \storeMap(\ell) \, \storeMap
= \Sem{M \Bind W} e' \, \storeMap 
= \GSext{(\Sem{W} e')}(\Sem{M} e') \, \storeMap
\]
But $x \not\in FV(W)$ implies $\Sem{W} e' = \Sem{W} e$, and we are done.
\end{proof}

Equalities (\ref{prop:true-eq-i})-(\ref{prop:true-eq-iii}) in Proposition \ref{prop:true-eq} are the {\em monadic laws} from Definition \ref{def:Cmonad}, Equalities (\ref{prop:true-eq-iv})-(\ref{prop:true-eq-v})
are instances of Equation (\ref{eq:gen-op}) at the beginning of this section.


\section{Intersection Type Assignment System}
\label{sec:imp-intersection}

Intersection types have been introduced in the '80 as an extension of Curry's simple type assignment system to the untyped $\lambda$-calculus, and
have been developed into a large family of systems, with applications to various $\lambda$-calculi, including call-by-value and linear ones.
Prominent results concern the characterization of normalizing and strongly normalizing as well as solvable terms, namely reducing to a head normal form.
It is then natural to think of intersection types as a tool to characterize convergent terms in $\lamImp$-calculus.

We extend the system in \cite{deLiguoroTreglia20} to type stores and operations $\textit{get}_{\ell}$ and $\textit{set}_{\ell}$,
by specializing the generic monad $T$ to the state and partiality monad.

\myparagraph{Types and subtyping}
An intersection type theory $\Th_A$ over a language of types $\Lang_A$ is an axiomatic presentation of a preorder $\leq_A$ over $\Lang_A$
turning the latter set into an inf-semilattice with top.

\begin{definition}\label{def:theories}
	An {\em intersection type theory}, shortly {\em itt}, is a pair $\Th_A = (\Lang_A, \leq_A)$ where $\Lang_A$, the {\em language} of 
	$\Th_A$, is a countable set of type expressions closed under $\Inter$, and $\omega_A\in \Lang_A$ is a special constant; $\leq_A$ 
	is a pre-order over $\Lang_A$ closed under the following rules:
	\[\begin{array}{c@{\hspace{1.5cm}}c@{\hspace{1.5cm}}c@{\hspace{1.5cm}}c}
		\alpha \leq_A \omega_A &
		\alpha \Inter \beta \leq_A \alpha &
		\alpha \Inter \beta \leq_A \beta &
		\prooftree
		\alpha \leq_A \alpha' \quad \beta \leq_A \beta'
		\justifies
		\alpha \Inter \beta \leq_A \alpha' \Inter \beta'
		\endprooftree
	\end{array}
	\]
\end{definition}
 
From the operational semantics, and especially Proposition \ref{prop:big-small}, and from the denotational semantics in the previous section, it is clear that we have four kinds of entities to type:
the values $V$ and $\Aget{\ell}(s)$, the store terms $s$, the computations $M$, and the configurations $(M, s)$. Therefore, we define four sorts of types.

\begin{definition}\label{def:typetheories}
	We define four sorts of types by mutual induction as follows:
	
	\[\begin{array}{rlll@{\hspace{0.5cm}}lllll}
		\Lang_D ~ \ni & \delta & ::= & \delta \to \tau \mid \delta \Inter \delta' \mid \omegaD   &\text{Value Types}\\ [2mm]
		\Lang_{\Store} ~ \ni
		         & \sigma & ::= & \tuple{\ell : \delta} \mid \sigma \Inter \sigma' \mid \omegaS 
		              & \text{Store Types} \\ [2mm]
		\Lang_{C} ~ \ni & \kappa & ::= & \delta \times \sigma \mid \kappa \Inter \kappa' \mid \omegaC  &\text{Computation Types}\\ [2mm]  
		\Lang_{\GS D} ~ \ni & \tau & ::= & \sigma \to \kappa \mid \tau \Inter \tau' \mid \omegaSD  &\text{Configuration Types}
	\end{array}\]
\end{definition}

We assume that $\Inter$ and $\times$ take precedence over $\to$ and that $\to$ associates to the right so that
$\delta \to \tau \Inter \tau'$ reads as $\delta \to (\tau \Inter \tau')$ and
$ \delta' \to \sigma' \to \delta'' \times \sigma''$ reads as $\delta' \to (\sigma' \to (\delta'' \times \sigma''))$. 

Since (closed) values are abstractions, types in $\LangD$ are arrow types from value (as the calculus is call-by-value)
to computation types. To the arrows, we add atoms plus intersections and a universal type for values, as we do for all type sorts.

Store types $\LangS$ associate value types to locations. 
As observed after Proposition \ref{prop:big-small}, computations are functions from stores to results, hence they are typed by arrow types
$\sigma \to \kappa$ where $\kappa$ is the type of results. Results are pairs of values and store terms so that types in 
$\LangC$ have the form $\delta \times \sigma$ and the intersection thereof, plus a universal type $\omegaC$ which is the type of all results and
the undefined one, that is the ``result'' of a diverging reduction.

Similarly to the system in \cite{deLiguoroTreglia20}, we introduce a subtyping relation $\leq_A$ for each sort $\Lang_A$, which is in the style
of the system in \cite{BCD'83}. Subtyping is a form of implication which, if types are interpreted as the extension of predicates over terms, is subset inclusion.
Such a choice has advantages and disadvantages: with subtyping, we can assign to terms just arrow, $\tuple{\ell:\delta}$ 
(a sort of record type) or product types, disregarding the cases of intersection and of
the trivial types $\omega$ in the type inference rules in Figure \ref{fig:store-lookup-config-assignment}, 
that otherwise jeopardize their form. Moreover, intersection types with subtyping do correspond to basic open subsets
of domain theoretic models, providing a ``logical semantics'' to the calculus in the sense of \cite{Abramsky'91}.
A drawback is that the type system is not syntax-directed, making proofs more elaborate. This is partly remedied by means of
the Generation Lemma \ref{lem:genLemma}.

\begin{definition}[Subtyping]\label{def:type-preorders}
For each sort $A = D,S,C,\GS D$, we define the itt $\Th_A = (\Lang_A, \leq_A)$ as the least one such that:
\begin{enumerate}
\item \label{def:type-preorders-1} $\omega_{D} \leq_{D} \omega_D \to \omega_{\GS D}$
\item \label{def:type-preorders-2} $(\delta \to \tau) \Inter (\delta \to \tau') \leq_{D} \delta \to (\tau \Inter \tau')$
\item \label{def:type-preorders-3} $\tuple{\ell : \delta}\Inter \tuple{\ell : \delta'}\leq_{S} \tuple{\ell : \delta\Inter \delta'}$
\item \label{def:type-preorders-5} $(\delta \times \sigma) \Inter (\delta' \times \sigma') \leq_{C} (\delta \Inter \delta' )\times (\sigma \Inter \sigma')$
\item \label{def:type-preorders-6} $\omega_{\GS D} \leq_{\GS D} \omega_S \to \omega_{C}$
\item \label{def:type-preorders-7} $(\sigma \to \kappa) \Inter (\sigma \to \kappa') \leq_{\GS D} \sigma \to (\kappa \Inter \kappa')$
\end{enumerate}
and that are closed under the rules:
\[\begin{array}{c@{\hspace{1cm}}c}
				
		\prooftree
		\delta' \leq_{D} \delta \quad \tau \leq_{\GS D} \tau'
		\justifies
		\delta \to \tau \leq_{D} \delta' \to \tau'
		\endprooftree 
		&
		\prooftree
		\delta \leq_{D} \delta' \quad \sigma \leq_S \sigma'
		\justifies
		\delta \times \sigma \leq_{C} \delta' \times \sigma'
		\endprooftree

		\\[6mm]
		
		\prooftree
		\sigma' \leq_{S} \sigma \quad \kappa \leq_C \kappa'
		\justifies
		\sigma \to \kappa \leq_{\GS D} \sigma' \to \kappa'
		\endprooftree
		&
		\prooftree
		\delta\leq_{D}\delta' 
		\justifies
		\tuple{\ell: \delta} \leq_{S} \tuple{\ell: \delta'}
		\endprooftree
		
	\end{array}
\]
\end{definition}

We avoid subscripts in the $\omega$ and the $\leq$ whenever clear from the context.
It is easily seen that all axioms of the shape $\varphi \leq \varphi'$ imply their reverse $\varphi' \leq \varphi$; when this is the case we write
$\varphi = \varphi'$ and say that such types are {\em equivalent}.
We use the abbreviations $\bigwedge_{i\in I}\varphi_i \equiv \varphi_1 \Inter \cdots \Inter \varphi_n$ if $I = \Set{1,\ldots,n}$, 
and $\bigwedge_{i\in \emptyset}\varphi_i \equiv \omega$.

We comment on some equalities and inequalities derivable from Definition \ref{def:type-preorders}. Obviously, the intersection $\Inter$ is
idempotent, commutative and associative, and $\varphi \Inter \omega = \varphi$.

The arrow types in $\Th_D$ and $\Th_{\GS D}$ are treated in the same way as in \cite{BCD'83} for ordinary intersection types.
In particular, we have that $\delta \to \omegaSD = \omegaD$ for all $\delta \in \Lang_D$, which implies $\omegaD\to \omegaSD = \omegaD$; 
a similar remark applies to $\Th_{\GS D}$. 

Types $\tuple{\ell : \delta} \in \Lang_S$ are covariant in $\delta$ w.r.t. $\leq_D$, therefore by \ref{def:type-preorders}.\ref{def:type-preorders-3},
\[\tuple{\ell : \delta}\Inter \tuple{\ell : \delta'}\leq \tuple{\ell : \delta\Inter \delta'} \leq \tuple{\ell : \delta}\Inter \tuple{\ell : \delta'}\]
hence they are equivalent.
Notice that $\omega_S \not\leq \tuple{\ell : \omega_D}$ hence the inequality $\tuple{\ell : \omega_D} < \omega_S$, which holds in any itt, is strict. The reason
for distinguishing among these types will be apparent in the definition of the type system; see also the remark below.

Product types $\delta \times \sigma \in \Lang_C$ are covariant in both $\delta$ and $\sigma$, therefore 
\ref{def:type-preorders}.\ref{def:type-preorders-5} implies 
$(\delta \times \sigma) \Inter (\delta' \times \sigma') = (\delta \Inter \delta' )\times (\sigma \Inter \sigma')$.

\begin{remark}\label{rem:omegaSdiverso}
In comparison to \cite{deLiguoroT21a}, the key difference is that in the cited paper there was no distinction among 
$\omegaS$ and $ \tuple{\ell : \omegaD}$ since they were equated in the theory. 
The present change essentially amounts to dropping such an equation. This reflects at the level of types
that the semantics of a store $s$ such as $\ell \not \in \dom{s}$ must be different than the semantics
of some $s'$ such that $\lkp{\ell}{s'} = V$ even if the meaning of $V$ is the bottom of $\DinftyS$,
for example when $V \equiv \lambda x. \Omega_c$; in fact, in such a case we have $\ell \in \dom{s'}$. This is further clarified by the subsequent definition and lemma, where the concept of the domain of a store is extended to store types.
\end{remark}

\medskip
The following definition will be useful in technical development, and more precisely in the side condition of rule $\setR$ of the typing system.

\begin{definition}\label{def:dom-sigma}
	For any $\sigma \in \Lang_S$ define $\dom{\sigma} \subseteq \Label$ by: 
	\[\begin{array}{rll}
		\dom{\tuple{\ell : \delta}} & = & 
			\Set{\ell} \\
		\dom{\sigma \Inter \sigma'} & = & \dom{\sigma} \cup \dom{\sigma'} \\ [1mm]
		\dom{\omegaS} & = & \emptyset \\
	\end{array}\]
\end{definition}

\begin{lemma}\label{lem:sigma-dom}
For all $\sigma,\sigma' \in \LangS$
\[\sigma \leqS \sigma' \Then \dom{\sigma} \supseteq \dom{\sigma'}\]
\end{lemma}

\begin{proof} By induction over the definition of $\leqS$.
\end{proof}

\begin{definition}[Type Assignment System]\label{def:type-system}
A {\em typing context} is a finite set $\Gamma = \Set{x_1 : \delta_1, \ \ldots , \ x_n : \delta_n}$ with pairwise distinct $x_i$'s and with $\delta_i \in \LangD$ for all $i = 1, \ldots, n$; with $\Gamma$ as before, we set
$\dom{\Gamma} = \Set{x_1, \ldots , x_n}$; finally, by $\Gamma, \, x : \delta$ we mean $\Gamma \cup \Set{x : \delta}$ for $x \not\in \dom{\Gamma}$.
Type {\em judgments} are of either forms:
\[ \Gamma \der V : \delta \qquad \Gamma \der M : \tau \qquad \Gamma \der s : \sigma \qquad \Gamma \der (M, s) : \kappa \]
The rules of the type assignment system are listed in Figure \ref{fig:store-lookup-config-assignment}.
\end{definition}

\myparagraph{Illustrating the type system w.r.t. the operational semantics}
Here we turn to a more concrete
understanding of rules in Figure \ref{fig:store-lookup-config-assignment}, looking at the operational semantics. The discussion is informal and includes 
some examples; the denotational interpretation and the technical analysis of the typing system w.r.t. the operational semantics are deferred to the next sections.

\medskip
Types are interpreted as predicates of terms expressing some property of them
and can be understood set theoretically: $\varphi \Inter \psi$ is the conjunction of $\varphi$ and $\psi$,
namely their intersection. If $\varphi \leq \psi$ then $\varphi$ implies $\psi$ and hence it is a subset of $\psi$. 
Finally {\em universal types} $\omega_A$ are interpreted as the union of all sets $\varphi \in \Lang_A$. 

The intended meaning of a typing judgment $\Gamma \der Q:\varphi$ is that $Q$ satisfies
the property $\varphi$, namely $Q \in \varphi$, under the hypothesis that $x$ satisfies $\Gamma(x)$ for all $x\in \dom{\Gamma}$. 
Types $\delta \in \LangD$ are properties of terms in $\ValTerm$;
since term variables range over values and are values themselves, in a typing context $\Gamma$ 
we only have typing
judgments of the form $x:\delta$ with $\delta \in \LangD$.

\begin{figure*}[htp]
{\bf Rules for $\omega$, intersection, and subtyping}

\bigskip
\[\begin{array}{c@{\hspace{1cm}}c@{\hspace{1cm}}c@{\hspace{1cm}}c}
&
\prooftree
	\vspace{3mm}
\justifies
	\Gamma \der Q : \omega
\using \omegaR
\endprooftree
&
\prooftree
	\Gamma \der Q : \varphi
	\quad
	\Gamma \der Q : \varphi'
\justifies
	\Gamma \der Q : \varphi \Inter \varphi'
\using \interR
\endprooftree
&
\prooftree
	\Gamma \der Q : \varphi
	\quad
	\varphi \leq \varphi'
\justifies
	\Gamma \der Q : \varphi'
\using \leqR
\endprooftree
\end{array}\]

\bigskip\bigskip

{\bf Rules for $\ValTerm$ and $\ComTerm$ terms}


\[\begin{array}{c@{\hspace{1cm}}c}
\prooftree
	\vspace{3mm}
\justifies
	\Gamma, x:\delta \der x : \delta
\using \varR
\endprooftree
&
\prooftree
	\Gamma, \, x:\delta \der M : \tau
\justifies
	\Gamma \der \lambda x.M : \delta \to \tau
\using \lambdaR
\endprooftree
\\ [6mm]
\prooftree
	\Gamma \der V:\delta
\justifies
	\Gamma \der \unit{V}: \sigma \to \delta \times \sigma
\using \unitR
\endprooftree
&
\prooftree
	\Gamma \der M : \sigma \to \delta' \times \sigma'
	\quad
	\Gamma \der V : \delta' \to \sigma' \to \delta'' \times \sigma''
\justifies
	\Gamma \der M \Bind V : \sigma \to \delta'' \times \sigma''
\using \bindR
\endprooftree
\end{array}\]

\bigskip
\[\begin{array}{c}
\prooftree
	\Gamma, \, x:\delta \der M : \sigma \to \kappa
\justifies
	\Gamma \der \get{\ell}{\lambda x.M} : (\tuple{\ell : \delta} \Inter \sigma) \to \kappa
\using \getR
\endprooftree
\\ [8mm]
\prooftree
	\Gamma \der V : \delta
	\quad
	\Gamma \der M :  (\tuple{\ell : \delta} \Inter \sigma)  \to \kappa
	\quad
	\ell \not\in \dom{\sigma} 
\justifies
	\Gamma \der \set{\ell}{V}{M} : \sigma \to \kappa
\using \setR
\endprooftree
\end{array}\]

\bigskip\bigskip

{\bf Rules for $\StoreSort$, $\LkpSort$ terms and configurations}

\bigskip
\[\begin{array}{c@{\hspace{1cm}}c@{\hspace{1cm}}c@{\hspace{1cm}}c}
&
\prooftree
	\Gamma \der V : \delta
\justifies
	\Gamma \der \Aset{\ell}(V, s) : \tuple{ \ell : \delta }
\using \updRa
\endprooftree
&
\prooftree
	\Gamma \der s : \tuple{ \ell' : \delta }
	\quad \ell \neq \ell'
\justifies
	\Gamma \der \Aset{\ell}(V, s) : \tuple{ \ell' : \delta }
\using \updRb
\endprooftree
	\\ [6mm]
&
\prooftree
	\Gamma \der s : \tuple{ \ell : \delta }
\justifies
	\Gamma \der \Aget{\ell}(s) : \delta
\using \lkpR
\endprooftree
&

\prooftree
	\Gamma \der M : \sigma \to \kappa
	\quad
	\Gamma \der s : \sigma
\justifies
	\Gamma \der (M, s) : \kappa
\using \confR
\endprooftree
\end{array}\]

\bigskip

\caption{Type assignment system for $\lamImp$}\label{fig:store-lookup-config-assignment}
\end{figure*}

\paragraph{Rules for $\omega$, intersection and subtyping}
Rules for $\omega_A$, intersection and subtyping are the same as in \cite{BCD'83}, instantiated
to the four kinds of a judgment of the present system. Universal types $\omega_A$
do not carry any information about their subjects, but for their sort, hence they 
are usually called {\em trivial types}  and
are assigned to any term of the appropriate kind.

As explained above, intersection and subtyping are understood set theoretically, which justifies
rules $\interR$ and $\leqR$. In particular, if $\varphi = \varphi'$ then $\Gamma \der Q:\varphi$ implies $\Gamma \der Q:\varphi'$.

\paragraph{Rules for $\StoreSort$ and $\LkpSort$ terms} 
Before illustrating the rules for values and computations, that are the core of the assignment system, let us discuss
the rules for typing store terms. 

Store types are in $\LangS$.
Since there is no rule to type $\emp$, the only possibility is to derive $\Gamma \der \emp : \omegaS$, for arbitrary 
$\Gamma$, up to type equivalence because of rule $\leqR$.
For the same reason  $\Gamma \der \Aget{\ell}(\emp) : \omegaD$ is the only possible typing of $\Aget{\ell}(\emp)$.

The intended meaning of typing $s$ by the type $\tuple{\ell: \delta}$ is that $\ell \in \dom{s}$ and $s$ associates some value $V$
of type $\delta$ to the location $\ell$, namely that $\Aget{\ell}(s) = V$ has type $\delta$.
Indeed, by Corollary \ref{cor:s-nf-decidable}, any $s \in \StoreSort$ can be equated to a term of the form:
\[\Aset{\ell_1}(V_1, \cdots \Aset{\ell_n}(V_n, \emp) \cdots)\]
with distinct locations $\ell_1, \ldots , \ell_n$. If $s \neq \emp$ then $n > 0$ and 
$\Gamma \der V_i:\delta_i$ for some $\delta_i$ and each $1 \leq i \leq n$,  so that
we can derive $\Gamma \der s : \tuple{ \ell_i : \delta_i }$ for all $i$ by either using rule
$\updRa$ or $\updRb$. Hence, by repeated use of $\interR$, the judgment
$\Gamma \der s: \bigwedge_{1 \leq i \leq n} \tuple{\ell_i : \delta_i}$ can be derived,
expressing that $s$ associates to each $\ell_i$ a value of type $\delta_i$. In summary,
the following rule is derivable:
\begin{align}\label{eq:fun-interp}
\prooftree
	\Gamma \der V_1 : \delta_1 \quad \cdots \quad \Gamma \der V_n : \delta_n
\justifies
	\Gamma \der \Aset{\ell_1}(V_1, \cdots \Aset{\ell_n}(V_n, \emp) \cdots) : \bigwedge_{1 \leq i \leq n} \tuple{\ell_i : \delta_i}
\endprooftree
\end{align}
Therefore, if $s$ has type $\bigwedge_{1 \leq i \leq n} \tuple{\ell_i : \delta_i}$
then $\Aget{\ell_i}(s) = V_i$ has type $\delta_i$, by rules $\leqR$ and $\lkpR$.

\paragraph{Rules for $\ValTerm$ and $\ComTerm$ terms}
Coming to the rules for values and computations, rule $\varR$ is obvious.
 About rule $\lambdaR$ we observe that
in our syntax abstractions have the shape $\lambda x.M$ where $M \in \ComTerm$, so that
they represent functions from values (the variable $x$ can be substituted only by terms in $\ValTerm)$ 
to computations. Also, abstractions
are values and by rule $\lambdaR$ are assigned 
functional types $\delta \to \tau \in \LangD$ if $\delta$ is the type of $x$ and
$\tau \in \LangSD$ is a type of $M$, as usual. 

Notice that even $\lam x. \Omega_c$, representing the everywhere undefined function, is a value; as we will see the only types of $\Omega_c$
are equivalent to $\omegaSD$, hence by rule $\lambdaR$ we can type $\lam x. \Omega_c$ by $\omegaD \to \omegaSD$ that is equivalent to
$\omegaD$ which indeed is the type of any value, not of computations.

Computations are assigned types $\tau \in \LangSD$; recall that these types are (intersections of) arrows $\sigma \to \kappa$, and
when $\kappa \neq \omegaC$ they are (intersections of) arrow types $\sigma \to \delta\times\sigma'$ 
for some store types $\sigma,\sigma'$ and value type $\delta$. 
As observed after Proposition \ref{prop:big-small}, we can see a computation $M$ as a (partial) function from stores to results.
If $M(s)$ is defined then there is a unique $(V,t)$ such that $M(s) = (V,t)$. 
Hence the intended meaning of the typing $M:\sigma \to \kappa$ is that $M$ is a function sending any store $s$ of type
$\sigma$ to a result of type $\kappa$, which is the content of rule $\confR$. As a consequence of our results, and especially of 
Theorem \ref{thr:subject-exp} and
Theorem \ref{thr:char-convergence},
we have that, if $M(s)$ is undefined, we cannot derive any arrow type $\sigma \to \kappa$ for $M$ with $\kappa \neq \omegaC$, therefore both $M$
and $(M,s)$ will be only typeable by rules $\omegaR$ and $\leqR$.

Rules $\unitR$ and $\bindR$ are instances of the
corresponding rules in pure computational $\lambda$-calculus in \cite{deLiguoroTreglia20}, specialized 
to the case of the state monad. The term $\unit{V}$ represents the trivial computation
returning $V$; as a function of stores it is such that $\unit{V}(s) = (V,s)$ behaving as the constant function w.r.t. the value $V$, 
and as the identity w.r.t. the store $s$; therefore  if $\delta$ is a type of $V$
any type $\sigma \to \delta \times \sigma$ can be assigned to $\unit{V}$ by rule $\unitR$.

In rule $\bindR$ the first premise says that applying $M$ to a store $s$ of type $\sigma$ yields a result $M(s) = (W,t)$ such that
$W:\delta'$ and $t:\sigma'$. We abbreviate this by saying  that $(W,t) = \unit{W}(t)$ has type $\delta' \times \sigma'$. 
In the second premise the value $V$ of type $\delta' \to \sigma' \to \delta'' \times \sigma''$ is the currified version
of a function of type $\delta' \times \sigma' \to \delta'' \times \sigma''$ (which is not a type in our formalism) sending
$(W,t)$ to some new result of type $\delta'' \times \sigma''$. By the operational semantics of the $\Bind$ operator,
we then conclude that for arbitrary $s$ of type $\sigma$, $(M \Bind V)(s)$  yields a result of type $\delta'' \times \sigma''$,
hence the typing $M \Bind V: \sigma \to \delta'' \times \sigma''$ in the conclusion of the rule.

According to rule $\getR$, the result of $(\get{\ell}{\lambda x.M})(s)$ has type
$\kappa$ if $s$ has type $\tuple{\ell:\delta} \Inter \sigma$, hence it has both type $\tuple{\ell:\delta}$ and $\sigma$. As we have seen above, if
$s$ has type $\tuple{\ell:\delta}$ then $\Aget{\ell}(s) = V$ for some value $V$ of type $\delta$.
Now, from the operational semantics we may argue that $(\get{\ell}{\lambda x.M})(s) = (M\Subst{V}{x})(s)$; 
hence if $s$ has type $\sigma$ then a sufficient condition to produce a result $M(s)$ of type $\kappa$
is to assume that $x$ in $M$ has type $\delta$,
as required in the premise of the rule.

Rule $\setR$ is the central and subtlest rule of the system. Reasoning as before, we argue that  \[(\set{\ell}{V}{M})(s) = M(\Aset{\ell}(V, s))\] and
we look for sufficient conditions for the result having type $\kappa$ whenever $s$ has type $\sigma$. Now if $V$ has type
$\delta$ as in the first premise, then $\Aset{\ell}(V, s)$ has type $\tuple{\ell:\delta}$; moreover,
it has also type $\sigma$, because $\ell \not\in \dom{\sigma}$, namely the side condition of the rule so that it suffices
that $M$ has type $\tuple{\ell:\delta} \Inter \sigma \to \kappa$ to conclude, which is the second premise.

\medskip
	The rules $\setR$ and $\getR$ are dual, in a sense. 
	Consider the derivation:
	\[
	\prooftree
	\Gamma \der V : \delta \qquad
	\prooftree
	\Gamma, \, x:\delta \der M : \sigma \to \kappa
	\justifies
		\Gamma \der \get{\ell}{\lambda x.M} : (\tuple{\ell : \delta} \Inter \sigma) \to \kappa
	\using \getR
	\endprooftree
	\qquad \ell \not\in \dom{\sigma}
	\justifies
	\Gamma \der \set{\ell}{V}{\get{\ell}{\lambda x.M}}: \sigma \to \kappa
	\using \setR
	\endprooftree
	\]
	Since $\tuple{\ell: \delta}$ is strictly less than $\omegaS$ for any $\delta$, it follows that 
	\[\omegaS \to \kappa < \tuple{\ell: \delta}\Inter \sigma \to \kappa\] 
	for any $\delta$ and $\sigma$. Indeed,  as a consequence of Theorem \ref{thr:char-convergence}, 
	we know that  $\Gamma \not\der \get{\ell}{\lambda x.M} : \omegaS\to \omegaD\times\omegaS$.
	
	\eject
	
	 In fact, $\get{\ell}{\lambda x.M}\not\Uparrow$ because $ (\get{\ell}{\lambda x.M}, \emp) $ is a blocked configuration. 
	But if the term $\get{\ell}{\lambda x.M}$ is placed in a suitable context, as in the conclusion of the derivation above, then the type $\omegaS\to \omegaD\times\omegaS$ is derivable for $ \set{\ell}{V}{\get{\ell}{\lambda x.M}} $ depending on the typing of $M$. 
	E.g. consider $M\equiv \Unit{x}$, one can take $\sigma = \omegaS$ and $\kappa=\omegaD \times\omegaS$; from the derivation above, we see that $M'= \set{\ell}{V}{\get{\ell}{\lambda x.\Unit{x}}}$ has type $\omegaS\to\omegaD \times\omegaS$. 
	On the other hand, $(M',\emp)$ converges:
	\[
	(M',\emp)\Red (\get{\ell}{\lam x.\Unit{x}}, \upd{\ell}{V}{\emp})\Red (\Unit{V}, \upd{\ell}{V}{\emp})
	\] 
	
\medskip
It remains to explain why the side condition $\ell \not\in \dom{\sigma}$ in rule $\setR$ is not only sufficient but also necessary.
By Lemma \ref{lem:sigma-dom} if $\sigma \leqS \tuple{\ell:\delta}$ then $\ell \in \dom{\tuple{\ell:\delta}} \subseteq \dom{\sigma}$; by contraposition
if $\ell \not\in \dom{\sigma}$ then 
there exists no $\delta \in \LangD$ such that $\sigma \leqS \tuple{\ell:\delta}$.
Now, if $\sigma = \bigwedge_i \tuple{\ell_i : \delta_i}$ then for all $i$ we have $\ell_i \neq \ell$. This means that from the side condition no information about
the value associated with $\ell$ by stores of type $\sigma$ is recorded in the type.
 
 Therefore the typing 
$M:(\tuple{\ell : \delta} \Inter \sigma)  \to \kappa$ in the second premise does not depend on any information about the value
of the store at $\ell$, but for the $\delta$ in $\tuple{\ell : \delta}$, which is the type of $V$ in the first premise.

Dropping the side condition we could have, say, 
$M:\tuple{\ell:\delta} \Inter \tuple{\ell:\delta'} \to \kappa$
where $\tuple{\ell:\delta} \Inter \tuple{\ell:\delta'} = \tuple{\ell:\delta \Inter \delta'}$, with both $\delta, \delta' \in \LangD$, and hence we could derive the type
$\tuple{\ell:\delta'} \to \kappa$ for $\set{\ell}{V}{M}$ in the conclusion. Now this says
that if $s$ has type $\tuple{\ell:\delta'}$ then $(\set{\ell}{V}{M})(s) = M(\Aset{\ell}(V,s))$ will have type
$\kappa$, notwithstanding that the value of $\Aset{\ell}(V,s)$ at $\ell$ is $V$, of which we only
know that it has type $\delta$, not $\delta'$.
We further illustrate the point in the next example.

\begin{example}\label{ex:typing-set-get}
A key issue in defining the system in Figure \ref{fig:store-lookup-config-assignment} is how to type terms like
$$M \equiv \set{\ell}{W}{ \set{\ell}{V}{ \get{\ell}{\lambda x.\unit{x} }  } }$$ from Example \ref{ex:overriding}, 
exhibiting the strong update property of the calculus. Let us abbreviate $N \equiv \set{\ell}{V}{ \get{\ell}{\lambda x.\unit{x} }  } $,
and suppose that $\der V: \delta$ and $\der W:\delta'$; we want to
derive $\der M: \sigma \to \kappa$ for suitable $\sigma$ and $\kappa$. 
Suppose that $\ell \not \in \dom{\sigma}$, then we have
\[
\prooftree
	\der V: \delta \quad 
		\prooftree
			\prooftree
				x : \delta \der x: \delta
			\justifies
				x : \delta \der \unit{x} : \sigma \to \delta \times \sigma
			\using \unitR
			\endprooftree
		\justifies
			\der  \get{\ell}{\lambda x.\unit{x} } : \tuple{\ell : \delta}  \Inter \sigma \to \delta \times \sigma	
		\using \getR
		\endprooftree
		\quad \ell \not \in \dom{\sigma}
\justifies
	\der N: \sigma \to \delta \times \sigma
\using \setR
\endprooftree
\]
Now taking $\kappa = \delta \times \sigma$ we have 
$\sigma \to \kappa \leq \tuple{\ell : \delta'} \Inter \sigma \to \kappa$, by
the fact that $\tuple{\ell : \delta'} \Inter \sigma \leq \sigma$ and the contravariance of $\to$ to its left.
From this and rule $\leqR$ we derive  
$\der N :\tuple{\ell : \delta'} \Inter \sigma \to \kappa$, hence we conclude
\[
\prooftree
\der W: \delta' \quad \der N: \tuple{\ell : \delta'} \Inter \sigma \to \kappa \quad \ell \not \in \dom{\sigma}
\justifies
\der M : \sigma \to \kappa
\using \setR
\endprooftree
\]

\medskip
Now, suppose that $\der s: \sigma$: hence $\der (M,s) : \kappa = \delta \times \sigma$.
From Example \ref{ex:overriding}, we know that
$\SmallStepStar{M}{s}{\unit{V}}{\Aset{\ell}(V, \Aset{\ell}(W, s) )}$; since we expect that types
are preserved by reduction (see Theorem \ref{thr:subject-red}), we must be able to type
$\der \unit{V} : \sigma \to \delta \times \sigma$, which is the conclusion of $\unitR$ from the hypothesis that $\der V: \delta$.

Finally, let us assume that $\sigma \neq \omegaS$ (otherwise the thesis is trivial). Then $\sigma = \bigwedge_{i}\tuple{\ell_i:\delta_i}$ so that
$\ell \not \in \dom{\sigma}$ implies that $\ell \neq \ell_i$ for all $i$. Therefore $\der s :  \tuple{\ell_i:\delta_i}$ for all $i$ by $\leqR$ and
\[
\prooftree
	\prooftree
		\der s :  \tuple{\ell_i:\delta_i}
	\justifies
		\der \Aset{\ell}(W, s) : \tuple{\ell_i:\delta_i}
	\using \updRb
	\endprooftree
\justifies 
	\der \Aset{\ell}(V, \Aset{\ell}(W, s) ) : \tuple{\ell_i:\delta_i}
\using \updRb
\endprooftree
\]
Now by repeated applications of $\interR$ we conclude 
$\der \Aset{\ell}(V, \Aset{\ell}(W, s) ) : \sigma$, as desired.

On passing, we observe that $\Aset{\ell}(V, \Aset{\ell}(W, s) ) = \upd{\ell}{V}{s}$ in the algebra of store terms, and that the latter has the same types as the former.
\end{example}

\section{Type interpretation}\label{sec:type-interpretation}

To further understand the type assignment system in the previous section, we provide the interpretation of types and typing judgments
in any $\lamImp$-model, and prove that the system is sound w.r.t. such interpretation. The system is also complete, as it is shown 
in \cite{deLiguoroT23} via
a filter-model construction.

Recall from Theorem \ref{thm:domain-equation} that $\DinftyS$ is actually a solution of the equation
$D \simeq [D \to \GS \,D]$, that can be unfolded in the system:
\[
\left\{
\begin{array}{rll}
D & = & [D \to \GS \,D] \\
S & = & (D_\bot)^\Label \\
C & = & (D \times S)_\bot \\
\GS\, D & = & [S \to C]
\end{array}
\right.
\]
Such a domain is a particular case of a $\lamImp$-model; hereby we define the type interpretation
maps w.r.t. any such model.

\begin{definition}\label{def:type-interp}
Let $\Dcat = (D, \GS, \SemD{\cdot},\SemSD{\cdot})$ be a $\lamImp$-model, where
$D\cong [D\to \GS D]$ via $(\Phi,\Psi)$, and let the domains
$D,S,C, \GS\,D$ be as above; then we define the maps $\Sem{\cdot}^A: \Lang_A \to \wp A$, for
$A =  D,S,C, \GS\,D$ as follows:
\begin{enumerate}
\item $\Sem{\omega_A}^A = A$ and $\Sem{\varphi\Inter\psi}^A = \Sem{\varphi}^A \cap \Sem{\psi}^A$
\item $\Sem{\delta \to \kappa}^D = \Set{d \in D \mid \forall d' \in \Sem{\delta}^D. \; \Phi(d)(d') \in \Sem{\kappa}^C }$
\item $\Sem{\tuple{\ell : \delta}}^S = \Set{\varsigma \in S \mid \ell \in \dom{\varsigma} \And \varsigma(\ell) \in \Sem{\delta}^D}$
\item $\Sem{\delta \times \sigma}^C = \Sem{\delta}^D \times \Sem{\sigma}^S$
\item $\SemSD{\sigma \to \kappa} = \Set{h \in \GS\,D \mid \forall \varsigma \in \Sem{\sigma}^S.\; h(\varsigma) \in \Sem{\kappa}^C}$
\end{enumerate}
\end{definition}

\begin{lemma}\label{lem:type-interpretation-leq}
For any $\varphi, \psi \in \Lang_A$
\[\varphi \leq_A \psi \Then \Sem{\varphi}^A \subseteq \Sem{\psi}^A\]
\end{lemma}

\begin{proof} By induction over the definition of the pre-orders $\leq_A$. 
\end{proof}

Before establishing the soundness theorem, we need to extend the term interpretation to store and lookup terms and to configurations.
In the next definition, we overload the notation for the interpretation maps $\Sem{\cdot}^A$. To avoid repetitions we fix a 
$\lamImp$-model $\Dcat$.

\begin{definition}\label{def:store-lkp-config-interpretation}
Define the mappings $\SemS{\cdot}: \StoreSort \to \Env \to S$, $\Sem{\cdot}^{D_\bot}: \LkpSort \to \Env \to D_\bot$
and $\SemC{\cdot}: (\ComTerm \times \StoreSort) \to \Env \to C$ by

\begin{enumerate}
\item $\SemS{\emp} e = \lambda \_. \bot$
\item $\SemS{\upd{\ell}{u}{s}} e = \SemS{s}[\ell \mapsto \Sem{u}^{D_\bot} e]$
\item $\Sem{V}^{D_\bot} e = \SemD{V} e$
\item $\Sem{\lkp{\ell}{s}}^{D_\bot} e = \SemS{s} e \, \ell$
\item $\SemC{(M,s)} e = (\SemSD{M} e)(\SemS{s} e)$
\end{enumerate}
\end{definition}

Given an environment $e \in \Env = \Var \to D$ to interpret term variables and a context $\Gamma$, we 
say that $\Gamma$ is consistent with $e$, written $e \models^\Dcat \Gamma$, 
if $e(x) \in \Sem{\Gamma(x)}^D$ for all $x \in \dom{\Gamma}$.

\begin{theorem}[Soundness of the Type System]\label{thm:soundness}\label{thr:soundness}
If $e \models^\Dcat \Gamma$ then
\begin{enumerate}
\item $\Gamma \der V : \delta \Then \SemD{V} e \in \Sem{\delta}^D$
\item $\Gamma \der M : \tau \Then \SemSD{M} e \in \SemSD{\tau}$
\item $\Gamma \der s : \sigma \Then \SemS{s} e \in \SemS{\sigma}$
\item $\Gamma \der (M, s) : \kappa \Then \Sem{(M,s)}^C e \in \Sem{\kappa}^C$
\end{enumerate}
\end{theorem}

\begin{proof} By induction over derivations, using Lemma \ref{lem:type-interpretation-leq} in case of $\leqR$ rule.
\end{proof}

\section{Type invariance}
\label{sec:imp-invariance}

A characteristic property of intersection types for ordinary $\lambda$-calculus is invariance both under reduction and expansion of the subject.
The analogous results are proved for the pure computational $\lambda$-calculus in \cite{deLiguoroTreglia20}. 
Here we extend such results to the present calculus w.r.t. the typing of algebraic operations.

A preliminary lemma states a fundamental property of subtyping of arrow types, both in $\Th_D$ and in $\Th_{\GS D}$:

\begin{lemma}\label{lem:continuity}
	Let $\tau \neq \omegaSD$ and $\kappa \neq \omegaC$, then:
	\begin{enumerate}
	\item $\bigwedge_{i\in I} (\delta_i \to \tau_i) \leqD \delta \to \tau \iff$ \\ 
		$\exists J \subseteq I. \; J \neq \emptyset \And 
		\delta \leqD \bigwedge_{j\in J} \delta_j \And \bigwedge_{j\in J} \tau_j \leqSD \tau$
	\item $\bigwedge_{i\in I} (\sigma_i \to \kappa_i) \leqSD \sigma \to \kappa \iff $ \\ $\exists J \subseteq I. \; J \neq \emptyset \And $
		$\sigma \leqS \bigwedge_{j\in J} \sigma_j \And \bigwedge_{j\in J} \kappa_j \leqC \kappa$
	\end{enumerate}
\end{lemma}

\begin{proof} By induction over the definition of $\leq$. \end{proof}

The next lemma is an extension of the analogous property of the
system in \cite{BCD'83}, also called Inversion Lemma in \cite[Section 14.1]{BarendregtDS2013}  because 
it is the backward reading of the typing rules.

\begin{lemma}[Generation lemma] \label{lem:genLemma}
Assume that $\delta \neq \omegaD$, $\sigma \neq \omegaS$, and $\tau \neq \omegaSD$, then:
\begin{enumerate}
\item \label{lem:genLemma-var} $\Gamma \der x: \delta \iff \Gamma(x) \leqD \delta$
\item \label{lem:genLemma-lambda}
		$\Gamma \der \lambda x.M: \delta  \iff$ 
		$\exists I, \delta_i, \tau_i.~  \forall i \in I.\; \Gamma ,x:\delta_i \der M:\tau_i \And$  $\bigwedge_{i\in I}\delta_i \to \tau_i \leqD \delta$
\item \label{lem:genLemma-unit}
		$\Gamma \der \unit{V} : \tau \iff$	\\ $ \exists I, \delta_i, \sigma_i.~  
		                          \forall i \in I.\;  \Gamma \der V:\delta_i \And $\\ $\bigwedge_{i\in I}\sigma_i \to \delta_i \times \sigma_i \leqSD \tau$
\item \label{lem:genLemma-bind}
		$\Gamma \der M\Bind V : \tau \iff$ \\
		$\exists I, \tau_i=\sigma_i \to \delta_i\times \sigma'_i, \delta''_i,\sigma''_i.~  
		                          \forall i \in I.\; \Gamma \der M:\sigma_i\to \delta''_i\times \sigma''_i \And$ \\ 
		$\Gamma\der V:\delta''_i\to\sigma''_i\to \delta_i\times \sigma'_i \And $ 
		$\bigwedge_{i\in I}\tau_i \leqSD \tau$
\item \label{lem:genLemma-get}
		$\Gamma \der \get{\ell}{\lam x.M} : \tau \iff$ \\
			$\exists I, \delta_i, \sigma_i, \kappa_i.~ \forall i\in I.\; 
			   \Gamma, x:\delta_i \der M: \sigma_i\to\kappa_i \And$ \\$ \bigwedge_{i\in I}(\tuple{\ell : \delta_i} \Inter \sigma_i \to \kappa_i ) \leqSD \tau$
\item \label{lem:genLemma-set}
		$\Gamma \der \set{\ell}{V}{M} : \tau \iff$ \\
		$\exists I, \delta_i,\sigma_i, \kappa_i.~ \forall i\in I.\; \Gamma \der M: \tuple{\ell : \delta_i} \Inter \sigma_i \to \kappa_i  \And $\\
		$\Gamma \der V:\delta_i \And \bigwedge_{i\in I}(\sigma_i \to \kappa_i)\leqSD \tau \And$ \\  $\ell\not\in \dom{ \bigwedge_{i\in I}\sigma_i}$
\item \label{lem:genLemma-7'} $\Gamma \der s:\sigma \iff $ \\ $\exists I, l_i, \delta_i .~ 
                                  \forall i\in I.\; \exists V_i.\; \Aget{\ell_i}{(s)}=V_i \And $ \\ 
                                  $ \Gamma \der V_i : \delta_i \And \sigma \leqS \bigwedge_{i\in I}\tuple{\ell_i : \delta_i}$
\item \label{lem:genLemma-lkp}$\Gamma\der\Aget{\ell}{(s)}:\delta \iff \exists I, \delta_i.~ 
                 \forall i\in I.\; \Gamma \der s:\tuple{\ell : \delta_i}\And \bigwedge_{i\in I} \delta_i \leqD \delta$
\item \label{lem:genLemma-upd}$\Gamma \der\Aset{\ell}(V, s):\sigma \iff $\\
		$\exists I, \delta_i, \ell_i, \delta.~ \forall i\in I.\; s=\bigwedge_{i\in I} \tuple{\ell_i : \delta_i} \mbox{ where }\ell_i\neq \ell \And$ \\
		$ \Gamma\der V:\delta \And \tuple{\ell : \delta} \Inter \bigwedge_{i\in I} \tuple{\ell_i : \delta_i}\leq \sigma$
\item \label{lem:genLemma-conf}$\Gamma\der (M,s):\kappa \iff $\\
		$\exists I, \sigma_i, \kappa_i.~ \forall i\in I.\; \Gamma \der M:\sigma_i \to \kappa_i \And$\\
		$\Gamma \der s: \sigma_i \And \bigwedge_{i\in I}\kappa_i\leq \kappa$		
\end{enumerate}
\end{lemma}

\begin{proof} The implications $\Leftarrow$ are immediate. 
To see the implications $\Then$ we reason by induction over the derivations, by distinguishing the cases
of the last rule. Parts \ref{lem:genLemma-var} and \ref{lem:genLemma-lambda} are the same as for ordinary intersection types and $\lambda$-calculus; 
part \ref{lem:genLemma-unit} is immediate by the induction hypothesis, hence we treat part \ref{lem:genLemma-bind}, \ref{lem:genLemma-get}, and \ref{lem:genLemma-set}, only.
	
	\textbf{4.}	If the last rule in the derivation of $\Gamma \der M\Bind V: \tau$ is $\bindR$ just take $I$ as a singleton set. 
	If it is $\leqR$ then the thesis follows immediately by induction and the transitivity of $ \leqSD$.
	Finally, suppose that the derivation ends by
	\[
	\prooftree
	\Gamma \der M\Bind V : \tau_1
	\quad
	\Gamma \der M\Bind V : \tau_2
	\justifies
	\Gamma \der M\Bind V : \tau_1 \Inter \tau_2
	\using \interR
	\endprooftree
	\] 
	and $\tau \equiv  \tau_1 \Inter \tau_2$. Then by induction we have
	$\exists I, \tau_i=\sigma_i \to \delta_i\times \sigma'_i, \delta''_i,\sigma''_i.~  
	\forall i \in I.\; \Gamma \der M:\sigma_i\to \delta''_i\times \sigma''_i \And$ $\Gamma\der V:\delta''_i\to\sigma''_i\to \delta_i\times \sigma'_i \And 
	\bigwedge_{i\in I}\tau_i \leqSD \tau_1$
	and
	$\exists J, \tau_j=\sigma_j \to \delta_j\times \sigma'_j, \delta''_j,\sigma''_j.~  
	\forall j \in J.\; \Gamma \der M:\sigma_j\to \delta''_j\times \sigma''_j \And$  $\Gamma\der V:\delta''_j\to\sigma''_j\to \delta_j\times \sigma'_j \And 
	\bigwedge_{i\in I}\tau_j \leqSD \tau_2$.
	From this, the thesis immediately follows by observing that
	\[
	\bigwedge_{i\in I}\tau_i \leqSD \tau_1 \And \bigwedge_{j\in J}\tau_j \leqSD \tau_2 \Then 
	\bigwedge_{i\in I}\tau_i  \Inter \bigwedge_{j\in J}\tau_j \leqSD \tau_1 \Inter \tau_2.
	\]
	
	\textbf{5.} If the last rule in the derivation of 	
	$\Gamma \der \get{\ell}{\lam x.M} : \tau$ is $(get)$ just take $I$ as a singleton set, as before. Also, case subsumption case is a routine proof by i. h. 
	Let's suppose the derivation ends by 
	\[
	\prooftree
	\Gamma \der \get{\ell}{\lam x.M} : \tau_1
	\quad
	\Gamma \der \get{\ell}{\lam x.M} : \tau_2
	\justifies
	\Gamma \der \get{\ell}{\lam x.M} : \tau_1 \Inter \tau_2
	\using \interR
	\endprooftree
	\] 
	Then by induction, we have:\\
	$\exists I, \delta_i, \sigma_i, \kappa_i.~ \forall i\in I.\; \Gamma, x:\delta_i \der M: \sigma_i\to\kappa_i \And$ 
	$ \bigwedge_{i\in I}(\tuple{\ell : \delta_i} \Inter \sigma_i \to \kappa_i ) \leqSD \tau_1$\\
	and 
	$\exists J, \delta_j, \sigma_j, \kappa_j.~ \forall j\in J.\; \Gamma, x:\delta_j \der M: \sigma_j\to\kappa_j \And$ 
	$ \bigwedge_{j\in J}(\tuple{\ell : \delta_j} \Inter \sigma_j \to \kappa_j ) \leqSD \tau_2$
	Since $\tau\neq \omegaSD$ then also $\tau_1,\tau_2\neq \omegaSD$, hence w.l.o.g. we can suppose that for $i=1,2$ there exist $\bar{\sigma_i}$ and 
	$\bar{\kappa_i}$ such that $\tau_i$ has the shape $\bar{\sigma_i}\to \bar{\kappa_i}$.
	
	By Part (2) of Lemma \ref{lem:continuity} we know that $\exists \bar{I} \subseteq I. \; \bar{I} \neq \emptyset \And $
	$\bar{\sigma_1} \leqS \bigwedge_{i\in \bar{I}} \sigma_i \wedge \tuple{\ell : \delta_i} \And \bigwedge_{i\in \bar{I} } \kappa_i \leqC \bar{\kappa_1}$
	and 
	$\exists \bar{J} \subseteq J. \; \bar{J} \neq \emptyset \And $
	$\bar{\sigma_2} \leqS \bigwedge_{j\in \bar{J}} \sigma_j  \wedge \tuple{\ell : \delta_j} \And \bigwedge_{j\in \bar{J} } \kappa_j \leqC \bar{\kappa_2}$
	By this, the thesis follows by observing that 
	$	\tuple{\ell: \bigwedge_{i\in \bar{I}} \delta_i \wedge \bigwedge_{j\in \bar{J}} \delta_j} \wedge \bigwedge_{i\in \bar{I}} \sigma_i
	 \wedge \bigwedge_{j\in \bar{J}} \sigma_j 
	 \to \bigwedge_{i\in \bar{I}} \kappa_i  \wedge \bigwedge_{j\in \bar{J}} \kappa_j \leqSD \tau_1 \wedge \tau_2 = \tau	$
	 
	 \textbf{6.} As in the previous parts, let's discuss just the case in which the last rule in the derivation is $\interR$: 
	 \[
	 \prooftree
	 \Gamma \der \set{\ell}{V}{M} : \tau_1
	 \quad
	 \Gamma \der \set{\ell}{V}{M}: \tau_2
	 \justifies
	 \Gamma \der \set{\ell}{V}{M} : \tau_1 \Inter \tau_2
	 \using \interR
	 \endprooftree
	 \] 
	 then, by induction we have: 
	 
	 $\exists I, \delta_i,\sigma_i, \kappa_i.~ \forall i\in I.\; \Gamma \der M: \tuple{\ell : \delta_i} \Inter \sigma_i \to \kappa_i  \And $
	 $\Gamma \der V:\delta_i \And \bigwedge_{i\in I}(\sigma_i \to \kappa_i)\leqSD \tau_1 \And$   $\ell\not\in \dom{ \bigwedge_{i\in I}\sigma_i}$
	 
	 $\exists J, \delta_j,\sigma_j, \kappa_j.~ \forall j\in J.\; \Gamma \der M: \tuple{\ell : \delta_j} \Inter \sigma_j \to \kappa_j  \And $
	 $\Gamma \der V:\delta_j \And \bigwedge_{j\in J}(\sigma_j \to \kappa_j)\leqSD \tau_2 \And$   $\ell\not\in \dom{ \bigwedge_{j\in J}\sigma_j}$.
	 
	 The thesis by similar use of Part (2) of Lemma \ref{lem:continuity} as in the previous point and by the fact that by Definition \ref{def:dom-sigma} we have that $\ell\not\in  \dom{ \bigwedge_{i\in I}\sigma_i} \cup \dom{ \bigwedge_{j\in J}\sigma_j}$, hence $\ell\not\in  \dom{ \bigwedge_{i\in I}\sigma_i \wedge \bigwedge_{j\in J}\sigma_j}$.
	 
\end{proof}

\begin{lemma}[Substitution and expansion] \label{lem:SubstLemma} 
%
\begin{enumerate}
\item If $\,\Gamma, x:\delta \der M:\tau$ and $\Gamma \der V:\delta$ then $\Gamma \der M\Subst{V}{x}: \tau$.
\item If $\,\Gamma \der M\Subst{V}{x}: \tau$ 
	then there exists $\delta \in \Lang_D$ such that:
	\[\Gamma \der V :\delta \quad\mbox{and}\quad \Gamma, x:\delta \der M:\tau\]
\end{enumerate}
\end{lemma}

\begin{proof}
Both parts are proved by induction over derivations, using Lemma \ref{lem:continuity} and Lemma \ref{lem:genLemma}.
\end{proof}

We are now in place to establish the type invariance property w.r.t. reduction and expansion:

\begin{restatable}[Subject reduction]{theorem}{SubRed}\label{thr:subject-red}
	\[ \Gamma \der (M, s) : \kappa \And \SmallStep{M}{s}{N}{t} ~ \Then ~ \Gamma \der (N, t) : \kappa\]
\end{restatable}

\begin{proof} Let us assume that $\kappa \neq \omegaC$ since the thesis is trivial otherwise. The proof is
	by induction over the definition of $(M,s) \Red (N,t)$, using Lemma \ref{lem:genLemma}. We treat the interesting cases.
	From the hypothesis $\Gamma \der (M,s): \kappa$ and by Lemma \ref{lem:genLemma}.\ref{lem:genLemma-conf}
	we have that there exist
	a finite set $I$ and types $\sigma_i, \kappa_i$ such that for all $i \in I$:
	\begin{enumerate}[(a)]
		\item \label{beta_c-1} $\Gamma \der M: \sigma_i \to \kappa_i$;
		\item \label{beta_c-2} $\Gamma \der s: \sigma_i$ 
		\item \label{beta_c-3}$\bigwedge_{i\in I} \kappa_i \leqC \kappa$.
	\end{enumerate}
	\begin{description}
		\item Case $\SmallStep{ \unit{V} \Bind (\lambda x.M') } {s} {  M'\Subst{V}{x} } {s}$:
		
		By (\ref{beta_c-1}), and using Parts \ref{lem:genLemma-unit} and \ref{lem:genLemma-bind} of Lemma \ref{lem:genLemma}, 
		for all $i\in I$ there is $J_i, \delta_{ij}, \delta'_{ij}, \delta''_{ij}, \sigma''_{ij}$ such that for all $j\in J_i$:
		\begin{enumerate}[(a)]\addtocounter{enumi}{3}
			\item  \label{beta_c-4} 
			$\Gamma \der V: \delta'_{ij}$ with $\bigwedge_{j \in J_i} \delta'_{ij} \leqD \delta_{ij}$, where 
			$\Gamma \der \unit{V}: \sigma_{ij}\to \delta'_{ij}\times \sigma_{ij}$
			\item  \label{beta_c-5} 
			$\Gamma\der \lam x.M':\delta_{ij}\to \tau_{ij}$ with $\bigwedge_{j \in J_i} \tau_{ij} \leqC \sigma_i \to \kappa_i$, where 
			$\tau_{ij}=\sigma\to\delta''_{ij}\times\sigma''_{ij}$.
		\end{enumerate}
		By applying Lemma \ref{lem:genLemma}.\ref{lem:genLemma-unit} to (\ref{beta_c-5}), for all $i\in I$, $ j\in J_i$, 
		there is a finite set $K_{ij}$, $\delta_{ijk}, \delta''_{ijk}, \sigma_{ijk}$ such that for all $k \in K_{ij}$:
		\begin{enumerate}[(a)]\addtocounter{enumi}{5}
			\item  \label{beta_c-6} $\Gamma, x:\delta_{ijk}\der M':\sigma_{ijk}\to\delta''_{ijk}\times\sigma''_{ijk}$ with 
			$\bigwedge_{k \in K}(\delta_{ijk}\to\sigma_{ijk}\to\delta''_{ijk}\times\sigma''_{ijk})\leqSD \delta_{ij}\to\tau_{ij}$.  
		\end{enumerate} 
		Set $\sigma_{ijk}\to\delta''_{ijk}\times\sigma''_{ijk}=:\tau_{ijk}$.
		In virtue of Lemma \ref{lem:continuity}, we may assume w.l.o.g. that there exists a not empty set $\overline{K}\subseteq K_{ij}$ such that
		$\delta_
		{ij} \leqD \bigwedge_{k \in \overline{K}} \delta_{ijk}$ and $\bigwedge_{k \in \overline{K}} \tau_{ijk} \leqSD \tau_{ij}$.
		
		By (\ref{beta_c-5}) we have: $\delta'_{ij}\leqD \delta_{ij} \leqD \delta_{ijk}\Then \Gamma \der V:\delta_{ijk}$.\\
		By (\ref{beta_c-6}) we have: $\bigwedge_{i\in I}\bigwedge_{j \in J_i}\bigwedge_{k \in \overline{K}}\leqSD \tau_{ij}\leqSD \sigma_i\to\kappa_i$.
		
		In conclusion, by Substitution Lemma \ref{lem:SubstLemma}, 
		$\Gamma \der M'\Subst{V}{x}: \sigma_i\to\kappa_i$, hence by $\confR$ 	
		$\Gamma \der( M'\Subst{V}{x},s):\kappa_i$; now by repeated applications of rule $\interR$ and 
		$\leqR$, $\Gamma \der( M'\Subst{V}{x},s):\kappa$.
		
		\item Case
		$\SmallStep {\set{\ell}{V}{M}}{s}{M}{\Aset{\ell}(V, s)}$:
		
		By applying Part \ref{lem:genLemma-set} of Generation Lemma to (\ref{beta_c-1}):
		
		$\forall i\in I.\;\exists J_i, \delta_{ij}, \sigma_{ij}, \kappa_{ij}$ such that for all $j\in J_i$:
		\begin{enumerate}
			\item \label{set-1}
			$\Gamma\der M':\tuple{\ell:\delta_{ij}}\Inter \sigma_{ij}\to \kappa_{ij}$
			\item \label{set-2}
			$ \Gamma\der V:\delta_{ij} $
			\item \label{set-3}
			$ \bigwedge_{j \in J_i}(\sigma_{ij}\to\kappa_{ij})\leqSD \sigma_i\to\kappa_i$
			\item \label{set-4}
			$\ell\not\in \dom{\bigwedge_{j \in J_i}\sigma_{ij}}$
		\end{enumerate}
		By Parts \ref{set-3} and \ref{set-4} of Lemma \ref{lem:genLemma}  and Lemma \ref{lem:continuity}: there exists $ \overline{J}\in J_i$
		\[ \sigma_i\leqS \bigwedge_{j \in \overline{J}}\sigma_{ij} \And
		\bigwedge_{j \in \overline{J}}\kappa_{ij}\leqC \kappa_i \]
		Moreover, by $\updRb$ and $\interR$: $\Gamma \der \Aset{\ell}(V,s): \sigma_i \Inter \tuple{\ell : \delta_{ij}}$. 
		We conclude by $\confR$ and $\leqR$.		
		
		\item Case 
		$ \SmallStep { \get{\ell}{\lambda x.M} } {s} { M \Subst{V}{x} } {s} $ where $ \Aget{\ell} (s) = V $:
		
		By applying Lemma \ref{lem:genLemma}.\ref{lem:genLemma-get} to \ref{beta_c-1}, for all $i\in I$ there exist 
		$J_i, \delta_{ij}, \sigma_{ij}, \kappa_{ij}$ such that for all $j\in J_i$:
		\[ \Gamma, x:\delta_{ij} \der M':\sigma_{ij}\to \kappa_{ij} \And \bigwedge_{j \in J_i}(\tuple{\ell :\delta_{ij}}\Inter \sigma_{ij}\to \kappa_{ij})
		\leqSD\sigma_i\to\kappa_i \]
		By Lemma \ref{lem:continuity} there exists $ \overline{J}\in J_i$
		\[ \sigma_i\leqS \bigwedge_{j \in \overline{J}}\tuple{\ell :\delta_{ij}}\Inter\sigma_{ij} \And \bigwedge_{j \in \overline{J}}\kappa_{ij}\leqC \kappa_i \]
		By Lemma \ref{lem:genLemma}.\ref{lem:genLemma-7'} we know that there exist at least one $V$ such that $\Aget{\ell}(s)=V$
		 and $\Gamma\der V:\delta_{ij}$. We conclude by applying Substitution Lemma \ref{lem:SubstLemma} and routine arguments. 
		
	\end{description}
	All other cases are immediate by Lemma \ref{lem:genLemma}.
\end{proof}

In order to prove the Subject expansion, we have to establish some properties of stores, relating store terms with their types.

\begin{lemma}\label{lem:restriction}
	~\hfill
	\begin{enumerate}
		\item \label{lem:restriction-i} $\ell \not\in \dom{\sigma} \Then [\Gamma \der s : \sigma \; \iff \; \Gamma \der \Aset{\ell}(V,s) : \sigma]$
		\item \label{lem:restriction-ii} $\Gamma \der \Aset{\ell}(V,s) : \sigma \And \sigma \leqS \tuple{\ell : \delta} \neq \omegaS \Then \Gamma \der V: \delta$
	\end{enumerate}
\end{lemma}

\begin{proof}
	\begin{description}
		\item (\ref{lem:restriction-i}) The if part is immediate by induction over the derivation of $\Gamma \der \Aset{\ell}(V,s) : \sigma$. For the only if part,
		when $\sigma = \omegaS$ the thesis follows by $\omegaR$. Otherwise let $\sigma = \bigwedge_{i\in I} \tuple{\ell_i:\delta_i}$. 
		Then $I \neq \emptyset$ and for all $i \in I$ we have $\delta_i \neq \omegaD$ and then $\ell_i \neq \ell$ because $\ell \not\in \dom{\sigma}$.
		Therefore, for all $i \in I$, $\Gamma \der s:  \tuple{\ell_i:\delta_i}$ by $\leqR$ and $\Gamma \der \Aset{\ell}(V,s) : \tuple{\ell_i:\delta_i}$ by 
		$\updRb$,
		and we conclude by $\interR$.
		
		\item (\ref{lem:restriction-ii}) By hypothesis $\sigma = \bigwedge_{i\in I} \tuple{\ell_i:\delta_i} \leqS \tuple{\ell: \delta}$, where w.l.o.g. 
		we assume that in
		$\bigwedge_{i\in I} \tuple{\ell_i:\delta_i}$ the $\ell_i$ are pairwise distinct. Then there exists exactly one $i' \in I$ such that 
		$\tuple{\ell_{i'}:\delta_{i'}} \leqS \tuple{\ell : \delta}$, so that $\ell_{i'} = \ell$ and $\delta_{i'} \leqD \delta$. Now $\Gamma \der \Aset{\ell}(V,s) : \sigma$
		implies that $\Gamma \der \Aset{\ell}(V,s) : \tuple{\ell_{i'}:\delta_{i'}} \equiv \tuple{\ell :\delta_{i'}}$ which is derivable only if 
		$\Gamma \der V: \delta_{i'}$.
		Then from $\delta_{i'} \leqD \delta$ we conclude by $\leqR$. 
		
	\end{description}
\end{proof}

\begin{proposition}\label{sigma-types-eq}
	$\Gamma \der s:\sigma \And \; \der s = t \Then \Gamma \der t:\sigma$
\end{proposition}

\begin{proof} By checking axioms in Definition \ref{def:store-terms-axioms}. The only interesting case is when
	$s \equiv \Aset{\ell}(V, \Aset{\ell}(W, s'))$ and $t \equiv \Aset{\ell}(V, s')$. If $\ell \not\in \dom{\sigma}$ then
	$\Gamma \der s: \sigma$ iff $\Gamma \der s':\sigma$ iff $\Gamma \der t:\sigma$ by Part (\ref{lem:restriction-i}) of Lemma \ref{lem:restriction}.
	
	If $\ell \in \dom{\sigma}$ then there exist $\delta \neq \omegaD$ and $\sigma'$ such that
	$\sigma = \tuple{\ell : \delta} \Inter \sigma'$ and $\ell \not \in \dom{\sigma'}$. By the above and that $\Gamma \der \sigma \leq \sigma'$,
	we have that $\Gamma \der s' : \sigma'$; by Part (\ref{lem:restriction-ii}) of Lemma \ref{lem:restriction}, $\Gamma \der V:\delta$, hence
	$\Gamma \der t \equiv \Aset{\ell}(V, s') : \tuple{\ell : \delta}$ by $\updRa$. By Part (\ref{lem:restriction-i}) of Lemma \ref{lem:restriction}
	$\Gamma \der t: \sigma'$, hence we conclude by $\interR$.
\end{proof}

We are now in place to prove the Subject expansion property of the typing system.

\begin{restatable}[Subject expansion]{theorem}{SubExp}\label{thr:subject-exp}
	\[ \Gamma \der (N, t) : \kappa \And \SmallStep{M}{s}{N}{t} ~ \Then ~ \Gamma \der (M, s) : \kappa\]
\end{restatable}

\begin{proof}
	The proof is by induction over  $(M,s) \Red (N,t)$, assuming that $\kappa \neq \omegaC$. The only interesting cases are the following.
	
\begin{description}
\item Case: $M\equiv \unit{V} \Bind (\lambda x. M')$ and $N\equiv M'\Subst{V}{x}$ and $s=t$. \\ 
	By the last part of Lemma \ref{lem:genLemma}, $\exists I, \sigma_i, \kappa_i.~ \forall i\in I. \;\Gamma \der N:\sigma_i \to \kappa_i \And$
	$\Gamma \der t: \sigma_i \And \bigwedge_{i\in I}\kappa_i\leq \kappa$. 
	By Lemma \ref{lem:SubstLemma} , for all $i\in I$ there exist $\delta_i$ such that $\Gamma \der V:\delta_i$ and
	$\Gamma, x: \delta_i \der M':\sigma_i\to \kappa_i$. Then $\Gamma \der \unit{V}:\sigma_i \to \delta_i\times \sigma_i$ by rule $(unit)$ and 
	$\Gamma \der \lambda x. M': \delta_i \to \sigma_i\to \kappa_i$ by rule $(\lambda)$.
	We conclude that  $\Gamma \der \unit{V} \Bind (\lambda x. M'): \sigma\to \kappa$ by rule $(\Bind)$ and $ \interR$.
	
\item Case $\SmallStep { \get{\ell}{\lambda x.M'} } {s} { M' \Subst{V}{x} } {s}$ where $ \Aget{\ell} (s) = V $. \\ As before, 
	by the last part of Generation Lemma \ref{lem:genLemma},  and by Lemma \ref{lem:SubstLemma}, 
	$\exists I, \sigma_i, \kappa_i.~ \forall i\in I.\;\Gamma \der N:\sigma_i \to \kappa_i \And$
	$\Gamma \der s: \sigma_i \And \bigwedge_{i\in I}\kappa_i\leq \kappa$, and for all $i\in I$ 
	there exist $\delta_i$ such that $\Gamma \der V:\delta_i$ and
	$\Gamma, x: \delta_i \der M':\sigma_i\to \kappa_i$.
	Since $\Gamma \der V:\delta_i$ we derive by rule $\updRa$ that $\Gamma \der s=\Aset{\ell}(V,s'):\tuple{\ell : \delta_i}$, 
	where we can assume w.l.o.g. that $s=\Aset{\ell}(V, s')$ where $l\not\in \dom{s'}$ by Definition \ref{def:store-terms-axioms}.
	By $\Gamma, x: \delta_i \der M':\sigma_i\to \kappa_i$ and by $\getR$ we obtain: 
	$ \Gamma \der \get{\ell}{\lambda x.M'} : (\tuple{\ell : \delta_i} \Inter \sigma_i) \to \kappa_i $. 
	By this and $\Gamma \der s: \tuple{\ell : \delta_i}\Inter \sigma_i$, we conclude that 
	$ \Gamma \der \get{\ell}{\lambda x.M'} : \kappa$ by $\confR$ and $\leqR$. 
	
\item	Case $\SmallStep {\set{\ell}{V}{M'}}{s}{ M'}{\Aset{\ell}(V,s)}$. \\
	By the last part of Lemma \ref{lem:genLemma}, $\exists I, \sigma_i, \kappa_i.~ \forall i\in I.\;\Gamma \der N:\sigma_i \to \kappa_i \And$
	$\Gamma \der \Aset{\ell}(V,s): \sigma_i \And \bigwedge_{i\in I}\kappa_i\leq \kappa$. 
	We distinguish two cases. Suppose $\ell\not\in \dom{\sigma_i}$, then 
	\[
	\prooftree
	\Gamma \der V : \omegaD
	\quad
	\Gamma \der M':  (\tuple{\ell : \omegaD} \Inter \sigma_i)  \to \kappa_i
	\justifies
	\Gamma \der \set{\ell}{V}{M'} : \sigma_i \to \kappa_i
	\using \setR
	\endprooftree
	\]
	
	By Lemma \ref{lem:restriction}.(i) we know that $\Gamma\der s:\sigma_i$. Hence, the thesis follows by $\confR$ and $(\Inter)$. \\
	Otherwise, suppose $\ell\in \dom{\sigma_i}$. By Lemma \ref{lem:restriction}.(ii), we have that there exist $\delta_i$ 
	such that $\sigma_i\leq \tuple{\ell :\delta_i}$ and $\Gamma\der V:\delta_i$. We can assume w.l.o.g. $\sigma_i=\tuple{\ell:\delta_i}\Inter \sigma'_i$ with $\ell\not\in\dom{\sigma'_i}$: 
	\[\begin{array}{rcll}
		\Gamma \der \Aget{\ell}(V,s):\sigma_i &\Then &\Gamma \der \Aget{\ell}(V,s):\sigma'_i &\\
		& \Then& \Gamma \der s: \sigma'_i & \text{since }\ell\not\in\dom{\sigma'_i} \mbox{ and by \ref{lem:restriction}.\ref{lem:restriction-ii}}\\
		& \Then & \Gamma \der (\set{\ell}{V}{M'}, s):\kappa_i& 
	\end{array}\]
	The thesis follows by application of $\confR$ and $\interR$.
\end{description}
\vspace{-2ex}
\end{proof}


\begin{example}\label{ex:invariance} In Example \ref{ex:reduction} we have seen that 
\[ (\set{\ell}{V}{\unit{W}} \, ; \, \get{\ell}{\lambda x.N} , s) \RedStar (N\Subst{V}{x}, \Aset{\ell}(V, s)) \]
where $M ; N \equiv M \Bind \mutelambda. N$. To illustrate subject reduction, 
let's first specialize the typing rule $\bindR$ to the case of $M;N$ as follows:

\[
\prooftree
	\Gamma \der M : \sigma \to \delta' \times \sigma'
	\quad
	\Gamma \der N : \sigma' \to \delta'' \times \sigma''
\justifies
	\Gamma \der M ; N : \sigma \to \delta'' \times \sigma''
\using \seqR
\endprooftree
\]

Now, consider the derivation of $ \Gamma \der \set{\ell}{V}{\unit{W}} \, ; \, \get{\ell}{\lambda x.N} : \sigma \to \kappa $
as in Figure \ref{fig:derivation}; therefore, assuming 
$\Gamma \der s : \sigma$ we conclude $\Gamma \der  (\set{\ell}{V}{\unit{W}} \, ; \, \get{\ell}{\lambda x.N} , s) : \kappa$
by rule $\confR$.

\begin{figure*}
{\scriptsize
$\begin{array}{c}
\prooftree
	\prooftree
		\Gamma \der V: \delta
		\quad
		\prooftree
			\Gamma \der W : \delta'
		\justifies
			\Gamma \der \unit{W} : (\tuple{\ell : \delta} \Inter \sigma) \to \delta' \times (\tuple{\ell : \delta} \Inter \sigma)
		\using \unitR
		\endprooftree
		 \ell \not \in \dom{\sigma}
	\justifies
		\Gamma \der \set{\ell}{V}{\unit{W}} : \sigma \to \delta' \times (\tuple{\ell : \delta} \Inter \sigma) 
	\using \setR
	\endprooftree
	\quad
	\prooftree
		\Gamma, x: \delta \der N : \sigma \to \kappa
	\justifies
		\Gamma \der \get{\ell}{\lambda x.N} : (\tuple{\ell : \delta} \Inter \sigma) \to \kappa
	\using  \getR
	\endprooftree
\justifies
	\Gamma \der \set{\ell}{V}{\unit{W}} \, ; \, \get{\ell}{\lambda x.N} : \sigma \to \kappa
\using \seqR
\endprooftree
\end{array}$}

\caption{Type derivation in Example \ref{ex:invariance}
}\label{fig:derivation}

\end{figure*}

The Substitution Lemma \ref{lem:SubstLemma}.1 implies that the following rule is admissible:
\[
\prooftree
	\Gamma \der V: \delta
	\quad
	\Gamma, x: \delta \der N : \sigma \to \kappa
\justifies
	\Gamma \der N\Subst{V}{x} : \sigma \to \kappa
\using 
\endprooftree
\]
On the other hand, we have that $\ell \not \in \dom{\sigma}$ implies that
$\sigma = \bigwedge_{i \in I} \tuple{\ell_i : \delta_i} $ and $\ell \neq \ell_i$ for all $i \in I$. Therefore,
by $\updRb$ and $\interR$, from the assumption that $\Gamma \der s$ we have
$\Gamma \der \Aset{\ell}(V, s) : \sigma$, hence
$\Gamma \der ( N\Subst{V}{x},  \Aset{\ell}(V, s) ) : \kappa$ by $\confR$.

Notice that the typing of $W$ does not take part
to the derivation of $\Gamma \der ( N\Subst{V}{x},  \Aset{\ell}(V, s) ) : \kappa$, which
corresponds to the fact that $W$ gets discarded in the reduction from $(\set{\ell}{V}{\unit{W}} \, ; \, \get{\ell}{\lambda x.N} , s)$ to
$( N\Subst{V}{x},  \Aset{\ell}(V, s) )$.
\end{example}


\section{The characterization theorem}
\label{sec:imp-characterization}

The main result of this section is Theorem \ref{thr:char-convergence} below, where convergent terms are characterized by
typability with a single type in the intersection type system from Section \ref{sec:imp-intersection}:
\[\forall M \in \ComTerm^0. ~ M \Downarrow ~ \iff ~ \der M : \omegaS \to \omegaD \times \omegaS\]
where by the precedence of $\times$ over $\to$ (see the paragraph after Definition \ref{def:theories})
the type in the statement above reads as $ \omegaS \to (\omegaD \times \omegaS)$.

To prove the only-if part it suffices the equivalence of the big-step and the small-step operational semantics, namely reduction,
established in Proposition \ref{prop:big-small}, and Theorems \ref{thr:subject-red} and \ref{thr:subject-exp}; more precisely, 
we need subject expansion when arguing that if $\Gamma \der (\unit{V},t): \kappa$ and $\SmallStepStar{M}{s}{\unit{V}}{t}$ then 
$\Gamma \der (M,s): \kappa$.

\begin{lemma}\label{lem:char-only-if}
\[\forall M \in \ComTerm^0. ~ M \Downarrow ~\, \Then~ \der M : \omegaS \to \omegaD \times \omegaS\]
\end{lemma}

\begin{proof} First note that
	\[\begin{array}{llll}
	M \Downarrow & \Then & (M, \emp) \Downarrow \\
	& \Then & \exists V, t. \; \BigStep{M}{\emp}{V}{t} \\
	& \Then & \exists V, t. \; (M, \emp) \RedStar (\unit{V},t) & \mbox{(*)}
	\end{array}\]
	by Proposition \ref{prop:big-small}.
	Now consider the type derivation:
	\[
	\prooftree
		\prooftree
			\prooftree
			\justifies
				\der V : \omegaD
			\using \omegaR
			\endprooftree
		\justifies
			\der \unit{V}: \omegaS \to \omegaD \times \omegaS
		\using \unitR
		\endprooftree
		\quad
		\prooftree
		\justifies
			\der t : \omegaS
		\using \omegaR
		\endprooftree
	\justifies
		\der (\unit{V},t) : \omegaD \times \omegaS
	\using \confR
	\endprooftree
	\]
By (*) and Theorem \ref{thr:subject-exp}, we have that $\der (M,\emp) : \omegaD \times \omegaS$ so that 
$\der M : \omegaS \to \omegaD \times \omegaS$ since
$\omegaD \times \omegaS \neq_C \omegaC$ and,
as observed in Section \ref{sec:imp-intersection}, $\der \emp : \omegaS$ is the only typing of $\emp$ in the empty context up to $=_{S}$.
\end{proof}

The proof of the if part of Theorem \ref{thr:char-convergence} is more difficult. We adapt to the present calculus
the technique of {\em saturated sets} used in \cite{Krivine-book'93} to denote certain sets of terms of ordinary $\lambda$-calculus, that
are closed by $\beta$-expansion. Saturated sets correspond to Tait's computable predicates (see \cite{Bakel-TCS'95} Definition 3.2.2) and are called ``stable'' 
in \cite{BarendregtDS2013}, \S 17.2.

\begin{definition}[Saturated sets]\label{def:comp-interp}

Let $\BotC$ be a new symbol and set $M(s) = \BotC $ if $(M,s)\Uparrow$. Then define $\sat{\cdot}{}$ as a map associating to each type a subset of closed terms, stores, or their combinations, depending on the their kinds, possibly including the symbol $\BotC$:
\begin{enumerate}


\item $\sat{\delta}{\II} \subseteq \ValTerm^{\,0}$ by $\sat{\omegaD}{\II} = \ValTerm^{\,0}$, and \\
	$\sat{\delta \to \tau}{\II} = \Set{ V \mid \forall \, W \in \sat{\delta}{\II}. \; \unit{W} \Bind V \in \sat{\tau}{\II} }$

\item \label{def:comp-interp-2}
	$\sat{\sigma}{\II} \subseteq \StoreSort^{\,0}$ by\\ $\sat{\omegaS}{\II} = \StoreSort^{\,0}$ and \\
	$\sat{ \tuple{\ell : \delta} }{\II} = \Set{s \mid \ell\in\dom{s}\And \exists V \in \sat{\delta}{\II}. \, \der \lkp{\ell}{s} = V}$

\item\label{def:comp-interp-3} $\sat{\kappa}{\II} \subseteq (\ValTerm^{\,0}  \times \StoreSort^{\,0}) \cup \Set{\BotC}$ by \\
	$\sat{\omegaC}{\II} = (\ValTerm^{\,0}  \times \StoreSort^{\,0}) \cup \Set{\BotC}$ and 
	$\sat{\delta \times \sigma}{\II} = \sat{\delta}{\II} \times  \sat{\sigma}{\II}$

\item\label{def:comp-interp-4} $\sat{\tau}{\II} \subseteq \ComTerm^{\,0}$ by $\sat{\omegaSD}{\II} = \ComTerm^{\,0}$ and \\
	$\sat{\sigma \to \kappa}{\II} = \Set{M  \mid  \forall s \in \sat{\sigma}{\II}. \ M(s) \in \sat{\kappa}{\II}}$

\item $\sat{\varphi \Inter \varphi'}{\II}  =  \sat{\varphi}{\II} \cap \sat{\varphi'}{\II}$ for $\varphi$ of any sort.
\end{enumerate}
\end{definition}

\begin{remark}\label{rem:type-interpretation}
Definition \ref{def:comp-interp} is close to Definition \ref{def:type-interp} of type interpretation;
the main difference is that types are just sets of (closed) terms and not subsets of domains.
	
	Differently from \cite{Krivine-book'93}, in the above definition, the interpretation of a type does not depend on a mapping $I$ 
	for type variables, since in this setting, the only atoms are the $\omega$'s.
	Nonetheless, the interpretation of types $\delta$ and $\tau$ are, in general, non trivial subsets of $\ValTerm^{\,0}$ and $\ComTerm^{\,0}$, respectively. 
	First, observe that $\sat{\omegaC}{\II}\neq \sat{\delta \times \sigma}{\II}$ for all $\delta$ and $\sigma$, 
	since $\BotC \not\in \sat{\delta \times \sigma}{\II}$. As a consequence, we have that 
	\[M\in \sat{\sigma\to \delta \times \sigma'}{\II} \Leftrightarrow \forall s \in \sat{\sigma}{\II}\,\exists (V,t)\in \sat{\delta\times\sigma'}{\II}.\, 
		(M,s)\Downarrow(V,t)\]
	For example, the interpretation of the type $\omegaD\to (\omegaS \to \omegaD \times\omegaS)$ is 
	\[\Set{V \mid \forall M, s. \, (\unit{W}\Bind V)(s) \in \sat{\omegaD\times\omegaS}{\II}}\] 
	therefore, $\lambda x.\unit{x} \in \sat{\omegaD\to (\omegaS \to \omegaD \times\omegaS)}{\II}$, but
	$\lam x.\Omega_c$ does not belong to such a set because $ (\unit{W}\Bind \lam x.\Omega_c)(s)  =\Omega_c (s) = \BotC$ for any $W$.
\end{remark}

The following lemma parallels Lemma \ref{lem:type-interpretation-leq}.

\begin{lemma}\label{lem:subseteq}
For any $\varphi, \psi \in \Lang_A$ of any kind, if $\varphi \leq_A \psi$ then 
$\sat{\varphi}{\II} \subseteq \sat{\psi}{\II}$.
\end{lemma}

\begin{proof}
By induction over the definition of type pre-orders. We only show that the property holds for the store types. First
\[ \sat{ \tuple{\ell : \delta} }{\II}\cap \sat{ \tuple{\ell : \delta'} }{\II} \subseteq \sat{ \tuple{\ell : \delta\Inter \delta'} }{\II}\]
If $s \in \sat{ \tuple{\ell : \delta} }{\II}\cap \sat{ \tuple{\ell : \delta'} }{\II}$ then 
$\ell \in \dom{s}$ and for certain $V \in \sat{\delta}{}$ and $W \in \sat{\delta'}{}$ we have $\lkp{\ell}{s} = V$ and $\lkp{\ell}{s} = W$. Here equality
means derivability from the axioms of Definition \ref{def:store-terms-axioms} therefore, by Lemma \ref{lem:store-unicity}, $V \equiv W$, 
namely they are the same value in $\sat{\delta}{\II} \cap \sat{\delta'}{\II} = \sat{\delta \Inter \delta'}{\II}$ as required.

Finally that  
\[ \delta \leq \delta' \Then \sat{\tuple{\ell:\delta}}{} \subseteq \sat{\tuple{\ell:\delta'}}{} \]
is immediate consequence of the induction hypothesis $\sat{\delta}{} \subseteq \sat{\delta'}{}$ and the Definition 
\ref{def:comp-interp}.\ref{def:comp-interp-2}.
\end{proof}

The saturated sets yield a sound interpretation of type judgments as stated in the next lemma:

\begin{lemma}\label{lem:compLemma}
Let $\Gamma \der M : \tau$ with $\Gamma = \Set{x_1 : \delta_1 , \ldots , x_n : \delta_n}$ and $M \in \ComTerm$. 
For any $V_1, \ldots , V_n \in \ValTerm^{\,0}$,
if $V_i \in \sat{\delta_i}{\II}$ for $i = 1, \ldots , n$ then
\[ M\Subst{V_1}{x_1} \cdots \Subst{V_n}{x_n} \in \sat{\tau}{\II} \]
\end{lemma}

\begin{proof}
By induction over the derivation of $\Gamma \der M : \tau$. We abbreviate 
$\overline{M} \equiv M\Subst{V_1}{x_1} \cdots \Subst{V_n}{x_n}$.
The thesis is immediate if the derivation consists just of $\varR$, and it follows immediately by the induction hypothesis if the derivation
ends by either $\lambdaR$,  $\unitR$ or $\interR$. In case it ends by $\leqR$ we use Lemma \ref{lem:subseteq}. The following are the remaining cases.

\begin{description}

\item Case $\bindR$: then $M \equiv M' \Bind V$, $\tau \equiv  \sigma \to \delta'' \times \sigma''$ and the derivation ends by:
	\[
	\prooftree
		\Gamma \der M' : \sigma \to \delta' \times \sigma'
		\quad
		\Gamma \der V : \delta' \to \sigma' \to \delta'' \times \sigma''
	\justifies
		\Gamma \der M' \Bind V : \sigma \to \delta'' \times \sigma''
	\endprooftree
	\]
	By induction $\overline{M'} \in \sat{\sigma \to \delta' \times \sigma'}{\II}$,
	hence for all $s \in \sat{\sigma}{\II}$ there exists a result $(W,s')\in \sat{\delta' \times \sigma'}{\II}$ such
	that $\BigStep{\overline{M'}}{s}{W}{s'}$. 
	Now
	\[\begin{array}{llll}
	& &\BigStep{\overline{M'}}{s}{W}{s'} \\
	& \Then &  (\overline{M'}, s) \RedStar  (\unit{W}, s') & \mbox{by Proposition \ref{prop:big-small}} \\
	& \Then &  (\overline{M'} \Bind \overline{V}, s) \RedStar  (\unit{W} \Bind \overline{V}, s') & \mbox{(*)} \\
	\end{array}\]
	Also by induction, we know that $\overline{V} \in \sat{\delta' \to \sigma' \to \delta'' \times \sigma''}{\II}$, therefore
	there exists $(U, t) \in \sat{\delta'' \times \sigma''}{\II}$ such that $\BigStep { \unit{W} \Bind \overline{V}} {s'} {U}{t} $; on the other hand:
	\[\begin{array}{llll}
	& & \BigStep { \unit{W} \Bind \overline{V}} {s'} {U}{t} \\
	& \Then &  (\unit{W} \Bind \overline{V}, {s'}) \RedStar (\unit{U}, t)   & \mbox{by Proposition \ref{prop:big-small}} \\
	& \Then &  (\overline{M'} \Bind \overline{V}, s) \RedStar  (\unit{U}, t) & \mbox{by  (*)} \\
	& \Then &  \BigStep{\overline{M'} \Bind \overline{V}}{s}{U}{t} & \mbox{by Proposition \ref{prop:big-small}}
	\end{array}\]
	Then we conclude that $\overline{M' \Bind V} \equiv \overline{M'} \Bind \overline{V} \in \sat{\sigma \to \delta'' \times \sigma''}{\II}$.

\medskip
\item Case $\getR$: then $M \equiv \get{\ell}{\lambda x.M'}$, $\tau \equiv (\tuple{\ell : \delta} \Inter \sigma) \to \kappa$ and the derivation ends by:
	\[
	\prooftree
		\Gamma, \, x:\delta \der M' : \sigma \to \kappa
	\justifies
		\Gamma \der \get{\ell}{\lambda x.M'} : (\tuple{\ell : \delta} \Inter \sigma) \to \kappa
	\endprooftree
	\]
	Assume $s \in \sat{\tuple{\ell : \delta} \Inter \sigma}{\II} =  \sat{\tuple{\ell : \delta}}{\II} \cap  \sat{\sigma}{\II}$. 
	Since $x:\delta$ is in the premise context, we have $\delta\in \LangD$ by context definition. By induction hypothesis for all $V\in\sat{\delta}{\II}$ we have
	 $ \overline{M'}\Subst{V}{x}\in \sat{\sigma\to\kappa}{\II}$.
	
	We have to prove that $(\overline{\get{\ell}{\lam x.M'}})(s)\equiv(\get{\ell}{\lam x.\overline{M'}})(s)\in \sat{\kappa}{\II}$ if 
	$s\in \sat{\tuple{\ell:\delta}\Inter\sigma}{\II}=\sat{\tuple{\ell:\delta}}{\II}\cap\sat{\sigma}{\II}$
	
	So we have to handle just two cases: $\delta$ is an intersection or not.
	The most interesting case is when $\delta$ is not an intersection. By hypothesis we have:
	\begin{enumerate}[(a)]
		\item\label{casoA} $s\in \sat{\tuple{\ell : \delta}}{\II}\Then \ell \in \dom{s}\And \exists V \in \sat{\delta}{\II}. \ \lkp{\ell}{s}= V$
		\item\label{casoB} $s\in\sat{\sigma}{\II}\And V \in \sat{\delta}{\II}\Then (\overline{M'}\Subst{V}{x})(s)\in \sat{\kappa}{\II}$
	\end{enumerate}
	 
	 Since $\lkp{\ell}{s}= V$, $(\get{\ell}{\lam x.\overline{M'}},s)\Red (\overline{M'}\Subst{V}{x}, s)$, and by Proposition \ref{prop:big-small} we have 
	 
	 $ (\get{\ell}{\lam x.\overline{M'}})(s) = (\overline{M'}\Subst{V}{x})( s) \in \sat{\kappa}{\II} $

	Now consider the remaining case when $\delta$ is an intersection, that is $\delta = \bigwedge_{i\in I} \delta_i$, where every $\delta_i$ is not an intersection. 
	By definition, we know that 
	$ V\in \sat{\delta}{\II} \Iff \forall i\in I. \ V\in \sat{\delta_i}{\II} $.
	Reasoning as in the previous case,  we conclude that $ \get{\ell}{\lam x. \overline{M'}}\in \sat{\tuple{\ell_i : \delta_i}\Inter \sigma\to \kappa}{\II}$. 
	But $\tuple{\ell : \delta_i}\Inter \sigma\to \kappa \leq \tuple{\ell: \bigwedge_{i \in I}\delta_i}\Inter \sigma\to \kappa\equiv \tuple{\ell : \delta}\Inter \sigma\to \kappa$. 
	We conclude by Lemma \ref{lem:subseteq}.
\medskip
\item Case $\setR$: then $M \equiv \set{\ell}{V}{M'}$, $\tau \equiv \sigma \to \kappa$ and the derivation ends by:
	\[
	\prooftree
		\Gamma \der V : \delta
		\quad
		\Gamma \der M' :  (\tuple{\ell : \delta} \Inter \sigma)  \to \kappa
		\quad
		\ell \not \in \dom{\sigma}
	\justifies
		\Gamma \der \set{\ell}{V}{M'} : \sigma \to \kappa
	\endprooftree
	\]
	In this case $\overline{M} \equiv \set{\ell}{\overline{V}}{\overline{M'}}$. Now, let $s \in \sat{\sigma}{\II}$ be arbitrary: then
	\[ (\set{\ell}{\overline{V}}{\overline{M'}}, s) \Red (\overline{M'}, \upd{\ell}{\overline{V}}{s}) \]
	so that  $\set{\ell}{\overline{V}}{\overline{M'}}(s) = \overline{M'}(\upd{\ell}{\overline{V}}{s})$. 
	From the side condition $\ell \not \in \dom{\sigma}$ and the hypothesis $s \in \sat{\sigma}{\II}$
	we deduce that for some set of indexes $J$, labels in $\dom{\sigma} = \Set{\ell_j \mid j \in J}$ and types $\delta_j\in \LangD$:
	\[  \sat{\sigma}{\II} = \sat{ \bigwedge_{j \in J} \tuple{\ell_j : \delta_j} }{\II} = \bigcap_{j \in J} \sat{ \tuple{\ell_j : \delta_j} }{\II} \]
	where $\ell \neq \ell_j$ for all $j \in J$. By Definition \ref{def:store-terms-axioms}.\ref{def:store-terms-axioms-b} we have, for all $j \in J$:
	\[ \Aget{\ell_j}(\Aset{\ell}(\overline{V}, s)) = \Aget{\ell_j}(s) \in \sat{\delta_j}{\II} \mbox{ and } \ell_j \in \dom{s}\]
	hence $\upd{\ell}{\overline{V}}{s} \in  \sat{\sigma}{\II}$ as well. On the other hand 
	\[ \lkp{\ell}{\upd{\ell}{\overline{V}}{s}} = \overline{V} \in \sat{\delta}{\II} \]
	by induction, so that $\upd{\ell}{V}{s} \in \sat{\tuple{\ell : \delta}}{\II}$ and hence
	\[ \upd{\ell}{V}{s} \in \sat{\tuple{\ell : \delta}}{\II} \cap \sat{\sigma}{\II} = \sat{\tuple{\ell : \delta} \Inter \sigma}{\II} \]
	It follows that
	\[ \set{\ell}{\overline{V}}{\overline{M'}}(s) = \overline{M'}(\upd{\ell}{\overline{V}}{s}) \in \sat{\kappa}{\II}\]
	since $\overline{M'} \in \sat{(\tuple{\ell : \delta} \Inter \sigma)  \to \kappa}{\II}$ by induction.
		
\end{description}
\end{proof}

\begin{theorem}[Characterization of convergence]\label{thr:char-convergence}
\[\forall M \in \ComTerm^0. ~ M \Downarrow ~ \iff ~ \der M : \omegaS \to \omegaD \times \omegaS\]
\end{theorem}

\begin{proof}
The only-if part is Lemma \ref{lem:char-only-if}. To show the if part,
by Lemma \ref{lem:compLemma} $\der M : \omegaS \to \omegaD \times \omegaS$ implies $M \in \sat{\omegaS \to \omegaD \times \omegaS}{\II}$,
where the typing context $\Gamma = \emptyset$ and no substitution is considered since $M$ is closed.

Recall that $\sat{\omegaS}{\II} = \StoreSort^{\,0}$ and $\sat{\omegaD \times \omegaS}{\II} = \ValTerm^{\,0} \times \StoreSort^{\,0}$.
Therefore, for any $s \in \StoreSort^{\,0}$ we have $M(s) \in \sat{\omegaD \times \omegaS}{\II}$,
namely there exist $V \in \ValTerm^{\,0}$ and $t \in  \StoreSort^{\,0}$
such that $\BigStep{M}{s}{V}{t}$, hence $M \Downarrow$.
\end{proof}

We finish this section with some final remarks and examples. 
Together with the Soundness Theorem \ref{thr:soundness},
Theorem \ref{thr:char-convergence} implies that if $M\!\Downarrow$ then for any $\lamImp$-model $(D, \GS)$
the denotation $\SemSD{M} e \neq \bot_{\GS D}$ for any environment $e$. Indeed, dropping the $e$ which is immaterial as $M$ is closed, we have
\[\begin{array}{llll}
 \der M : \omegaS \to \omegaD \times \omegaS & \Then & \SemSD{M} \in \Sem{\omegaS \to \omegaD \times \omegaS} \\
 & \Then & \forall \storeMap \in \Sem{\omegaS}.\; \SemSD{M}\,\storeMap \in \Sem{\omegaD \times \omegaS} = D \times S \\
 & \Then & \forall \storeMap \in S.\,\SemSD{M}\,\storeMap \neq \bot_C \\
 & \Then & \SemSD{M} \neq \bot_{\GS D}
\end{array}
\]
where remember that $S = (D_\bot)^\Label$ and $C = (D \times S)_\bot$, hence $\bot_C \not \in D \times S$, and that
$\bot_{\GS D}  \in [S \to C]$ is the everywhere undefined function over $S$, namely $\bot_{\GS D}\; \storeMap = \bot_C$ for all
$\storeMap \in S$. 

Not surprisingly, $\Omega_c$ has no type strictly less
than $\omegaC$ by Theorem \ref{thr:char-convergence}. Indeed, $\Omega_c \!\Uparrow$, since $(\Omega_c, s) \to (\Omega_c, s)$
for any store term $s$. This corresponds to the fact that $\SemSD{\Omega_c}\, e\, \storeMap$ for any environment $e$ and $\storeMap \in S$,
namely $\SemSD{\Omega_c} = \bot_{\GS D}$.

However, the converse implication does not hold, because for example $\SemSD{\get{\ell}{\lambda x.\unit{x} }} \neq \bot_{\GS D}$ since
$\SemSD{\get{\ell}{\lambda x.\unit{x} }} \storeMap \neq \bot_C$ if and only if $\ell \in \dom{\storeMap}$, hence not for all stores. In fact, $\get{\ell}{\lambda x.\unit{x} } \Uparrow$ and this is the case of any (closed) $B$ yielding a blocked configuration $(B, s)$ for some store term $s$. By Theorem \ref{thr:char-convergence} we also derive
that $\der \get{\ell}{\lambda x.\unit{x} } : \sigma \to \kappa$ for some non-trivial $\kappa \neq \omegaD \times \omegaS$ only if $\ell \in \dom{\sigma}$,
for which it suffices to take $\sigma = \tuple{\ell : \omegaD}$. 

By this last remark we may argue that, taking $N \equiv  \unit{x}$ in 
$M \equiv \set{\ell}{V}{\unit{W}} \, ; \, \get{\ell}{\lambda x.N}$ from Figure \ref{ex:invariance}, we see that $M$
can be typed by $\omegaS \to \kappa$ for some non-trivial $\kappa$ even if $\ell \not \in \dom{\omegaS}$, since
the otherwise non-convergent term $\get{\ell}{\lambda x.\unit{x} }$ now has type $\tuple{\ell:\delta} \to \kappa$, showing
the dependency from a store defined on $\ell$, which is provided by the term $\set{\ell}{V}{\unit{W}}$. For the very same
reason the term $\get{\ell}{\lambda x.\unit{x}} \, ; \,  \set{\ell}{V}{\unit{W}}$ is not typeable by 
$\omegaS \to \kappa \leq \omegaS \to \omegaD \times \omegaS$ as 
it diverges.



\section{Discussion and Related Work}
\label{sec:imp-Related}

The present work is the full and revised version of 
\cite{deLiguoroT21a}.
But for the operators $\textit{get}_{\ell}$ and $\textit{set}_{\ell}$, the calculus syntax is the same as in \cite{deLiguoroTreglia20}, where we considered a pure untyped computational $\lambda$-calculus, namely without operations nor constants. The reduction relation in Section \ref{sec:imp-operational} is strictly included in that one considered there, which is
the compatible closure of the monadic laws from \cite{Wadler-Monads}, oriented from left to right. In contrast, here we just retain rule $\betaR$
and take the closure under rule $\BindR$; as a consequence the reduction is deterministic.
Nonetheless, as shown in Proposition \ref{prop:true-eq} the monadic laws are consistent with the convertibility relation induced by the reduction of $\lamImp$.

The algebraic operators $\textit{get}_{\ell}$ and $\textit{set}_{\ell}$ come from Plotkin and Power \cite{PlotkinP02,PlotkinP03,Power06}.
The algebra of store terms is inspired to \cite{PlotkinP02}, where the store monad in \cite{Moggi'91} is
generated by the update and lookup operations. Our construction, however, does not perfectly match with Plotkin and Power's one, because here the set $\Label$ of locations is infinite.

We have borrowed the notation for $\textit{get}_{\ell}$ and $\textit{set}_{\ell}$ from the ``imperative $\lambda$-calculus'' in chapter 3 of \cite{amsdottorato9075},
where also a definition of the convergence predicate of a configuration to a result is considered. Such a definition is a particular case of
the analogous notion in \cite{LagoGL17,LagoG19} for generic algebraic effects.
In these papers it is stated in semantic terms, while we preferred the syntactical treatment
in the algebra of store terms.

The type system is the same as in \cite{deLiguoroT21a}, but for the crucial strict inequalities 
$\tuple{\ell:\omegaD} < \omegaS$ which were erroneously postulated as equalities in the conference paper.
The effect of such an error was that even a blocked term may have type 
$\omegaS \to \omegaD \times \omegaS$, contradicting the characterization theorem.
The main changes in the present version are therefore in the definition of saturated sets in 
Definition \ref{def:comp-interp} and then in the proof of Lemma \ref{lem:compLemma}; besides their intrinsic complexity
they now convey a much clearer intuition, where $\BotC$ represents a divergent computation.

To the present paper we have added a denotational interpretation of both terms and types; as a matter of fact, types have been derived from the semantics, as shown in \cite{deLiguoroT23}. The reason 
for adding the denotational interpretations here is 
to motivate and to help the understanding of the type system and of the interpretation of the saturated set itself. The exposition of this part is short and self-contained. 
The study of the semantics of the calculus and of the type system can be found in \cite{deLiguoroT23}.

With respect to the pre-existing literature on imperative and effectful $\lambda$-calculi, 
let us mention Pfenning's work \cite{Davies-Pfenning'00}, which is further discussed in 
\cite{DR07}. In Pfenning's and others' paper, intersection types are added to the
Hindley-Milner type system for ML to enhance type expressivity and strengthen polymorphism. However, the resulting system is unsound, which forces the restriction
of intersection types to values and the loss of certain subtyping inequations which are crucial in any type system inspired to \cite{BCD'83}, including the present one. The issue is due to reference types in ML, where the type of a location is its ``intrinsic'' type in the sense of Reynolds \cite{Reynolds'00}.
In contrast, our store types are predicates over the stores, namely ``extrinsic'' types, 
telling what are the meanings of the values associated to the locations in the store.
Indeed our system enjoys type invariance under reduction and expansion, 
which does not hold for Pfenning's system.

A further line of research which seems relevant to us concerns 
the type and effect systems, which have
been introduced in \cite{GiffordL86} and pursued 
in \cite{TalpinJ94}. 
In the insightful paper \cite{WadlerT03} a type and effect judgement $\Gamma \der e:A, \varepsilon$ is translated into
and ordinary typing $\Gamma \der e: T^\varepsilon A$, where $T$ is a monad.
This has fostered the application of type systems with monadic types to static analysis for code transformation in \cite{BentonKHB06,BentonKBH09},
but also raised the question of the semantics of the types $T^\varepsilon A$.

In the papers by Benton and others, the semantics of monadic types with effect decorations is given in terms of PERs that are preserved by
read and write operations. Such semantics validates equations that do hold under assumptions about the effects of the equated programs;
e.g. only pure terms, neither depending on the store nor causing any mutation, can be evaluated in arbitrary order, 
or repeated occurrences of the same pure code can be replaced by a reference where the value of the code is stored after computing it just once.

Such properties are nicely reflected in our types: if $\lambda x.M$ has type $\delta \to \sigma \to \delta' \times \sigma$, and $\dom{\sigma}$
includes all the $\ell$ occurring in $M$, then we know that the function represented by $\lambda x.M$ is pure, as $\sigma$ is non-trivial and $\lambda x.M$ leaves the store unchanged; similarly if
$M:  \sigma \to \delta \times \sigma$ and $N:\sigma \to \delta' \times \sigma$ then for any $P:\sigma \to \kappa$ both
$M;N;P$ and $N;M;P$ will have the same types, and hence the same behaviour. In general, this suggests how to encode regions with store types
in our system.

As we wrote in \cite{deLiguoroT23}, an interesting development of our system is to consider non-idempotent intersection types,
also called quantitative types, where $\tau \neq \tau \Inter \tau$ for any $\tau$, for which see the survey \cite{BucciarelliKV17}.
The type system we have treated here is instead a classical intersection type system namely idempotent;
in the time elapsed between \cite{deLiguoroT21a} and the present writing, a quantitative type system inspired to the present one
has been introduced in \cite{AlvesKR23}  for a call-by-value $\lambda$-calculus with global store to capture exact measures of time and space in the reduction of terms. 

Finally, in \cite{GavazzoTV24} the authors have explored how intersection type systems can be used in the wider case of algebraic operations of
various kind, and in connection with abstract relational reasoning, hence going beyond the particular case of side-effects.

\section{Conclusion}
In this paper, we presented a type assignment system to study the semantics of an imperative computational $\lambda$-calculus equipped
with global store and algebraic operators. 
The system defines the semantics of the untyped calculus, and we obtained a type-theoretic characterization of convergence.

We see the present work as a case study, albeit relevant in itself, toward a type-theoretic analysis
of effectful $\lambda$-calculi based on Moggi's idea of modelling effects by means of monads.



\end{document}